\documentclass[10pt]{article}

\usepackage{amsmath,amsfonts,amsthm,amssymb}
\usepackage{mathtools}
\usepackage{fullpage}

\usepackage{pdfpages,graphicx}

\usepackage{hyperref}
\usepackage[capitalise]{cleveref}
\usepackage{paralist}
\usepackage{xspace}
\usepackage{caption}

\usepackage[ruled,linesnumbered,vlined,boxed,commentsnumbered]{algorithm2e}
\SetAlgoLined\DontPrintSemicolon

\DeclareMathOperator*{\dist}{dist}

\DeclareMathOperator*{\num}{\mathfrak{num}}

\def\mun{\smash{\stackrel{\raisebox{2.5pt}{$\longleftarrow$}}{\smash{\num}}}}

\newcommand{\NN}{\ensuremath{\mathbb{N}}}

\newcommand{\ALG}{\textsf{ALG}\xspace}
\newcommand{\ALGI}{\textsf{ALG'}\xspace}
\newcommand{\Ex}[1]{\ensuremath{\mathbb{E}[#1]}}

\newtheorem{theorem}{Theorem}
\newtheorem{definition}[theorem]{Definition}
\newtheorem{lemma}[theorem]{Lemma}
\newtheorem{corollary}[theorem]{Corollary}
\newtheorem{claim}[theorem]{Claim}
\newtheorem{fact}[theorem]{Fact}
\newtheorem{remark}[theorem]{Remark}

\def\epsilon{\ensuremath{\varepsilon}}
\newcommand{\eps}{\ensuremath{\epsilon}\xspace}
\newcommand{\EPS}{\ensuremath{\mathcal{E}}\xspace}

\newcommand{\COMMENTED}[1]{{}}
\newcommand{\junk}[1]{{}}

\newcommand{\Artur}[1]{{}}
\renewcommand{\Artur}[1]{{{\footnote{$\cal A.C.$ --- \textcolor[rgb]{0.50,0.00,0.00}{#1}}}}}
\renewcommand{\Artur}[1]{\textcolor{blue}{Artur: #1}}
\renewcommand{\Artur}[1]{\footnote{\textcolor{blue}{Artur: #1}}}
\newcommand{\Stefan}[1]{{}}
\renewcommand{\Stefan}[1]{{{\footnote{Stefan --- \textcolor[rgb]{0.50,0.00,0.00}{#1}}}}}
\renewcommand{\Stefan}[1]{\textcolor{red}{Stefan: #1}}

\title{\textbf{Testing Depth First Search Numbering}\thanks{A preliminary version of the paper appeared in \emph{Proceedings of the 33rd European Symposium on Algorithms (ESA)}, pages 76:1–-76:15, Warsaw, Poland, September 15--17, 2025.}}

\author{Artur Czumaj\thanks{Department of Computer Science and DIMAP, University of Warwick, Coventry, UK, and University of Cologne, Germany. Email: A.Czumaj@warwick.ac.uk. Research partially supported by the Centre for Discrete Mathematics and its Applications (DIMAP), by a Weizmann-UK Making Connections Grant, and by an IBM Award%, and by the Key Profile Area (KPA) ``Intelligent Methods for Earth System Sciences'' at the University of Cologne
.}
    \and
Christian Sohler\thanks{University of Cologne, Germany. Email: csohler@uni-koeln.de.}
    \and
Stefan Walzer\thanks{Karlsruhe Institute of Technology, Germany. Email: stefan.walzer@kit.edu. Research funded by the Deutsche Forschungsgemeinschaft (DFG, German Research Foundation) --- 465963632.}}

\date{}

\begin{document}

%-------------------------------------------------------------------------------------------------------

\maketitle

%-------------------------------------------------------------------------------------------------------

\begin{abstract}
\emph{Property Testing} is a formal framework to study the computational power and complexity of sampling from combinatorial objects. A central goal in standard graph property testing is to understand which graph properties are testable with sublinear query complexity. Here, a graph property $P$ is testable with a sublinear query complexity if there is an algorithm that makes a sublinear number of queries to the input graph and accepts with probability at least $2/3$, if the graph has property $P$, and rejects with probability at least $2/3$ if it is $\varepsilon$-far from every graph that has property $P$.

%\emph{Property Testing} is a formal framework to study the computational power and complexity of sampling from combinatorial objects such as graphs and functions. A central goal in standard graph property testing is to understand which graph properties are testable with sublinear query complexity. Here, a graph property $P$ is testable with a sublinear query complexity if there is an algorithm that makes a sublinear number of queries to the input graph and accepts with probability at least $2/3$, if the graph has property $P$, and rejects with probability at least $2/3$ if it differs in  more than an $\varepsilon$-fraction of its description size from every graph that has property $P$.

In this paper, we introduce a new variant of the bounded degree graph model. In this variant, in addition to the standard representation of a bounded degree graph, we assume that every vertex $v$ has a unique label $\num(v)$ from $\{1, \dots, |V|\}$, %, i.e., the assignment of labels to vertices is a bijection.
%In order to facilitate this setting,
and in addition to the standard queries in the bounded degree graph model, we also allow a property testing algorithm to query for the label of a vertex (but \emph{not} for a vertex with a given label).

Our new model is motivated by certain graph processes such as a DFS traversal, which assign consecutive numbers (labels) to the vertices of the graph. We want to study which of these numberings can be tested in sublinear time. As a first step in understanding such a model, we develop a \emph{property testing algorithm for discovery times of a DFS traversal} with query complexity $O(n^{1/3}/\eps)$ and for constant $\eps>0$ we give a matching lower bound.
\end{abstract}

%-------------------------------------------------------------------------------------------------------

\section{Introduction}
\label{sec:Intro}

\emph{Depth-first search (DFS)} is one of the most useful and frequently used algorithmic primitives in graph algorithms.
The DFS algorithm has been known for well over a century \cite{Lucas1882,Tarry1895} as a technique for threading mazes and has been widely used in the design of graph and network algorithms since the late 1950s.
DFS is a basic tool to traverse a graph in a structured way and is central to solve textbook problems such as connectivity, topological sorting \cite{T76}, determining strongly connected components in directed graphs \cite{Tarjan72}, biconnected components in undirected graphs \cite{Tarjan72}, and to test planarity \cite{HT74}.
Because of its versatility and usefulness for solving many graph problems, DFS has become one of the most important tools in the design of algorithms for graphs, taught in the first year of many computing science study programs (see, e.g., the classical textbooks \cite{CLRS22,KT13}).
One of the most useful combinatorial properties of DFS is the \emph{DFS numbering of vertices}, which is \emph{the order in which a DFS algorithm discovers all vertices of the input graph} \cite{Tarjan72}. (See also Algorithms~\ref{alg:DFS-numbering-def} and \ref{alg:DFS-numbering-def-rec}, where $N(v)$ denotes the set of neighbors of vertex $v$.)

%------------------------------------------------------------------------

\junk
{ % such formatting is taking too much space
\def\dfsVisit{\textsc{dfs-visit}}
\begin{figure}[h]
\centerline{
\begin{minipage}[t]{0.47\textwidth}
\begin{algorithm}[H]
\caption{DFS numbering of $G$}
\label{alg:DFS-numbering-def}%\small
\KwIn{
    $G = (V,E)$
}
\KwOut{
    %permutation
    $\num: V \rightarrow \{1,\dots,|V|\}$
}
    \BlankLine
    mark all vertices as undiscovered\;
    $t \leftarrow 1$ \tcp{\small smallest unused ID}
    \For(\tcp*[h]{\small arbitrary order}){$v \in V$}{
        \If{$v$ is undiscovered}{
            \dfsVisit($v$)\;
        }
}
\end{algorithm}
\end{minipage}\qquad
\begin{minipage}[t]{0.53\textwidth}
\begin{algorithm}[H]%[htbp]
\caption{\dfsVisit($v$)}
\label{alg:DFS-numbering-def-rec}%\small
    mark $v$ as discovered\;
    $\num(v) \leftarrow t$\;
    $t \leftarrow t + 1$\;
    \For(\tcp*[h]{\small arbitrary order}){$u \in   N(v)$}{
        \If{$v$ is undiscovered}{
            \dfsVisit($v$)\;
        }
    }
\end{algorithm}
\end{minipage}
}
\end{figure}
}

%------------------------------------------------------------------------

\def\dfsVisit{\textsc{DSFVisit}}
\begin{figure}[h]\small
\centerline{
\begin{minipage}[t]{0.505\textwidth}
\begin{algorithm}[H]
\caption{DFS numbering of $G$}
\label{alg:DFS-numbering-def}%\small
\KwIn{
    undirected graph $G = (V,E)$
}
\KwOut{
    %permutation
    $\num: V \rightarrow \{1,\dots,|V|\}$
}
    \BlankLine
    mark all vertices as undiscovered\;
    $t \leftarrow 1$ \tcp{\small smallest unused ID}
    \For(\tcp*[h]{\small arbitrary order}){$v \in  V$}{
        \lIf{$v$ is undiscovered}{
            \!\!\!\dfsVisit($v$)
        }
}
\end{algorithm}
\end{minipage}\qquad
\begin{minipage}[t]{0.505\textwidth}
\begin{algorithm}[H]%[htbp]
\caption{\dfsVisit($v$)}
\label{alg:DFS-numbering-def-rec}%\small
    mark $v$ as discovered\;
    $\num(v) \leftarrow t$\;
    $t \leftarrow t + 1$\;
    \For(\tcp*[h]{\small arbitrary order}){$u \in  N(v)$}{
        \lIf{$v$ is undiscovered}{
            \!\!\!\dfsVisit($v$)
        }
    }
\end{algorithm}
\end{minipage}
}
\end{figure}

%------------------------------------------------------------------------

\junk{ % another formatting option
\def\dfsVisit{\textsc{dfs-visit}}
\begin{figure}[h]
\centerline{
\begin{algorithm}[H]
\caption{DFS numbering of $G$}
\label{alg:DFS-numbering-def}%\small
\KwIn{
    undirected graph $G = (V,E)$
}
\KwOut{
    %permutation
    $\num: V \rightarrow \{1,\dots,|V|\}$
}
    \BlankLine
    mark all vertices as undiscovered\;
    $t \leftarrow 1$ \tcp{\small smallest unused ID}
    \For(\tcp*[h]{\small arbitrary order}){$v \in  V$}{
        \lIf{$v$ is undiscovered}{
            \dfsVisit($v$)
        }
}
\end{algorithm}
}

\centerline{
\begin{algorithm}[H]%[htbp]
\caption{\dfsVisit($v$)}
\label{alg:DFS-numbering-def-rec}%\small
    mark $v$ as discovered\;
    $\num(v) \leftarrow t$\;
    $t \leftarrow t + 1$\;
    \For(\tcp*[h]{\small arbitrary order}){$u \in  N(v)$}{
        \lIf{$v$ is undiscovered}{
            \dfsVisit($v$)
        }
    }
\end{algorithm}
}
\end{figure}
}

%------------------------------------------------------------------------

We remark that sometimes, e.g., in the widely used textbook by Cormen et. al.\ (see Chapter~20.3 in \cite{CLRS22}), the DFS numbering is also used together with \emph{finishing numbers} (cf. \cref{sec:FIN-numbering}). We also remark that the labeling is not unique and depends on the ordering in which the vertices are traversed in the for-loops of Algorithms~\ref{alg:DFS-numbering-def} and~\ref{alg:DFS-numbering-def-rec}.

\begin{definition}[\textbf{DFS numberings}]
\label{def:DFS-numbering}
Let $G = (V,E)$ be a labeled undirected graph on $n$ vertices. %($G$ might be either undirected or directed).
A labeling $\num: V \rightarrow \{1,\dots,n\}$ is called a \emph{DFS numbering of $G$} if DFS gives rise to it for some order of processing of the vertices in the for-loops of Algorithms~\ref{alg:DFS-numbering-def} and \ref{alg:DFS-numbering-def-rec}.
\end{definition}

Given a wide applicability of DFS and DFS numbering, a natural problem is to verify whether a given graph is correctly DFS numbered, i.e., the labels assigned to the vertices are a DFS numbering for some ordering of vertices and neighbors. The latter problem is easy to solve by simulating a DFS using the given numbering and locally verifying that the DFS is performed correctly. However, is it also possible to \emph{approximately} verify whether the graph is properly DFS numbered using a sampling based algorithm? An answer to this question sheds some light on how the local structure of a DFS numbered graph (as implicitly given by the sample distribution) is related to the global DFS numbering.

We will study this question in the \emph{framework of Property Testing} (see, e.g., the monographs \cite{BY22,Goldreich17}). Property Testing provides a formal setting to study approximate decision problems. The goal is to distinguish objects that satisfy the tested property from those that are \eps-far from every object that satisfies the property. Here, \eps-far means that one has to modify more than an \eps-fraction of the object's description to obtain an object that has the property. A property testing algorithm requires some form of sampling access to the object.

In this paper, we will introduce a variant of the standard bounded-degree graph model \cite{GR02} that in addition to the standard setting, allows also \emph{label queries}. We will assume that there exists some labeling that assigns unique labels from the set $\{1, \dots, |V|\}$ to the vertex set $V$, i.e., the labeling is a bijection $\num: V \rightarrow \{1, \dots, |V|\}$. We say that a graph with maximum degree $d$ is \eps-far from a DFS numbered graph, if we have to modify (insert or delete) more than $\eps |V|$ edges to obtain a graph that is correctly DFS numbered\footnote{While one often uses the bound of $\eps d |V|$ edges instead of $\eps |V|$ edges, in the setting when $d=O(1)$ these terms differ only by a constant factor and as such there are essentially indistinguishable.}%
%
%\Artur{We will want to have $\eps d n$ rather than $\eps n$ \dots this will require some changes in the proofs.}
.
We do not allow to modify labels, though we notice that allowing to modify labels would also be a valid model (see \cref{subsubsec:property-testing-access-model} for some discussion). Our motivation to not consider it here was merely to stick as closely to the standard testing model as possible.

In this new framework, we study a fundamental problem of \emph{whether the input labeled undirected graph is properly DFS numbered} or \emph{it is \eps-far from having a valid DFS numbering}. Our main result is a tight complexity bound for this task (for constant $d$ and $\eps$). First, in \cref{section:DFS-lb}, we show that any tester for the DFS numbering requires $\Omega(n^{1/3})$ queries, and then, we complement this bound in \cref{section:DFS-ub} with an algorithm that tests DFS numbering with $O(n^{1/3}/\eps)$ queries.

The proof of the lower bound (\cref{thm:DFS-lb} in \cref{section:DFS-lb}) is by constructing two families $(G_n)_{n \in \NN}$ and $(B_n)_{n \in \NN}$ of \emph{good} and \emph{bad} random labeled graphs for which we will show that distinguishing between these families is necessary for any DFS tester. Then we will show that distinguishing between these families requires $\Omega(n^{1/3})$ queries.

The proof of the upper bound (\cref{thm:DFS-ub} in \cref{section:DFS-ub}) relies on a characterization of labelings that are \eps-far from a valid DFS numbering. The characterization is described in a form of some conflicts between the labelings and we design two algorithms detecting such conflicts: one %working well
for local conflicts and one %working well
for global conflicts. By combining these two algorithms we will obtain an algorithm that tests DFS numbering with $O(n^{1/3}/\eps)$ queries.

%-------------------------------------------------------------------------------------------------------

\subsection{Related work}

Property Testing was introduced by Rubinfeld and Sudan \cite{RS96} and first studied in the dense graph setting by Goldreich and Ron \cite{GGR98}. Constant time testability in the dense graph model has been fairly well understood \cite{AFNS09} and is closely related to the regularity lemma. In our paper we build upon the bounded degree graph model, as introduced by Goldreich and Ron \cite{GR02}. First results in this model included testers for properties such as connectivity, $k$-connectivity and being Eulerian \cite{GR02}; bipartiteness \cite{GR98} and cycle freeness \cite{CGRSSS14} can be tested in the bounded degree graph model in $O(\frac{\sqrt{n}}{\eps^{O(1)}})$ time. A series of papers proved that every hyperfinite property of bounded degree graphs is testable in constant time \cite{BSS10,CSS09,HKNO09,NS13}. Also, every constant time testable property in bounded degree graphs is either finite or contains a hyperfinite property \cite{FPS19}. Furthermore, it is known that every first-order logic property of bounded degree graph with an $\exists \forall$ quantification is constant time testable while some properties with a $\forall \exists$ quantification are not \cite{AKP21}. 
%Kumar et al.\  \cite{KSS19,KSS21} gave the first polynomial (in $d$ and $\eps^{-1}$) efficient property tester for minor-closed properties as well as an efficient partition oracle. The same authors also gave a one-sided error tester for minor-closed properties with $O(\frac{n^{1/2+o(1)}}{\eps^{O(1)}})$ query complexity \cite{KSS18}.
A number of other questions have been studied in the bounded degree graph model such as testability using proximity oblivious testers \cite{GR11}.
Further works in the bounded degree graph model include, e.g., testability using proximity oblivious testers \cite{GR11}.

%-------------------------------------------------------------------------------------------------------

\section{Preliminaries: The model}

%For simplicity of the presentation, we focus our exposition on undirected graphs and the discussion about directed graphs is deferred to \cref{section:directed-DFS}.

In this paper (except for \cref{section:directed-DFS}), let $G = (V,E)$ denote an undirected graph with $n = |V|$ vertices and $m = |E|$ edges.
Let $N(v)$ denote the set of \emph{neighbors of $v$}, i.e., $N(v) : = \{u \in V: \{v,u\} \in E\}$, and
%We use $d$ to the denote an upper bound for the \emph{maximum degree $G$} and assume that $d = O(1)$.
$d$ be an upper bound on the maximum number of edges incident to any vertex in $V$.
We will assume that $G$ is a labeled graph, i.e., there is a bijection (permutation) $\num: V \rightarrow [n]$ that assigns numbers to the vertices of $G$ and we use $[n]:= \{1,\dots,n\}$.

%-------------------------------------------------------------------------------------------------------

\subsection{Bounded degree graph model for \emph{labeled} graphs}
\label{subsubsec:property-testing-access-model}

In this section we describe how our algorithm can access the input graph $G = (V,E)$ of maximum degree $d$ with vertex labeling $\num$. Our model is a slight modification of the standard bounded degree graph model \cite{GR02} that in addition \emph{allows access to labels}. As usual in this model, we assume that $V=[n]$ and $n$ as well as $d$ are given to the algorithm. The algorithm can ask the following two types of queries and receives an answer in constant time:
\begin{itemize}
\item \textbf{Neighbor queries:} for every vertex $v \in V$ and every $1 \le i \le d$, one can query the $i$th neighbor of vertex $v$.
\item \textbf{Label queries:} for every vertex $v \in V$, one can query the label of $v$, $\num(v)$.
\end{itemize}
Observe that we allow access to vertices and their neighbors, and we can check the label $\num(v)$ of any vertex $v$,
\textbf{\emph{but we have no access to vertices through their labels $\num$}}.
In particular, \emph{to access vertex $v$ with a given label $\num$, we have to query the oracle until such vertex is returned}.
This is a special feature distinguishing our model from the standard model used in graph property testing and making the problem challenging: we have a labeled graph, but we cannot access the graph (its vertices) through the labels! Instead, the only way to access the graph is by taking any vertex $v \in V = [n]$ and either querying its label $\num(v)$ (via the label query) or querying its $i$th neighbor (via a neighbor query).

The \emph{query complexity} of a testing algorithm is the number of oracle queries.

A property testing algorithm (in short, \emph{property tester}) for DFS numbered graphs is an algorithm that has access to an input graph as described above and that accepts the input with probability at least $2/3$, if $G$ is a graph with $\num$ being a valid DFS numbering.
The algorithm rejects $G$ with probability at least $2/3$ if $G$ is \eps-far from having a valid DFS numbering $\num$ according to the following definition. (The algorithm developed in this paper is a \emph{property tester with one-sided error}, i.e., it always accepts, if $\num$ is a DFS numbering.)

\begin{definition}[\textbf{\eps-far from DFS numberings in bounded degree graphs}]
\label{def:eps-far-DFS-numbering-max-deg}
Let $G = (V,E)$ be an undirected graph on $n$ vertices with maximum degree at most $d$ % ($G$ might be either undirected or directed)
and let $\num: V \rightarrow [n]$ be a bijection that assigns labels to $V$. We say the labeled graph $G$ \emph{$\num$ is \eps-far from having a valid DFS numbering}, if one has to modify\footnote{Modification of the edges means edge insertions and deletions, i.e., we require $|E \triangle E'| > \eps n$, where $\triangle$ is the symmetric difference.} more than $\eps n$ edges in $G$ to obtain a graph $G' = (V,E')$ of maximum degree at most $d$ for which $\num$ is a valid DFS numbering.
\end{definition}

\paragraph{\emph{Implicitly labeled graphs.}}
To simplify the presentation, we will often assume $\num(v) = v$. However, \emph{we will not use this knowledge in the algorithm} as the model prevents us from querying $\num^{-1}$. Our tester will only ask for a random vertex or for a vertex that occurred as the neighbor of a previously queried vertex.

%-------------------------------------------------------------------------------------------------------

\subsubsection{Further thoughts about the model: Modifying the labels}
\label{subsubsec:discussion-about-model-editing-labels}
Observe that in general, it might be natural to consider a revised definition of a labeling $\num$ being \eps-far from a valid DFS numbering, where while defining graph $G'$ in
%Definitions \ref{def:eps-far-DFS-numbering}, \ref{def:eps-far-DFS-numbering-max-deg}, and \ref{def:eps-far-DFS-numbering-max-deg-connected},
\cref{def:eps-far-DFS-numbering-max-deg}, in addition to edge modifications, one would allow also for modifications of the labels\footnote{We use the following revised definition: a labeling \emph{$\num$ is \eps-far from a valid DFS numbering} of $G$ if for any graph $G' = (V,E')$ of maximum degree at most $d$ with a valid DFS numbering $\num'$ we have $|E \triangle E'| + |\{ v \in V: \num(v) \ne \num'(v)\}| \ge \eps n$.}. We observe that for bounded degree graphs, adding the labels modifications does not change the problem significantly.

\begin{lemma}
\label{lemma:two-models-editing-labels}
Let $\num$ be a permutation of $V$ to $\{1, \dots, n\}$. For a given bounded degree graph $G = (V,E)$, a labeling $\num$ is \eps-far from a valid DFS numbering according to the edges modifications if and only if $\num$ is $\Theta(\eps)$-far from a valid DFS numbering according to the edges and labels modifications.
\end{lemma}

\begin{proof}
Observe that the number of modifications of the edges to obtain a valid DFS numbering is not smaller than the number of modifications of the edges and the labels to obtain a valid DFS numbering. Therefore, if for a given bounded degree graph $G$, a labeling $\num$ is \eps-far from a valid DFS numbering according to the edges and labels modifications then $\num$ is also \eps-far from a valid DFS numbering according to the edges modifications.

For the other direction, if $\num$ is \eps-far from a valid DFS numbering according to the edges modifications then we can simulate modification of the labels by modification of the edges: to assign a given label to a vertex we just remove all its incident edges and add all edges incident to the vertex with the label sought. (Observe that a vertex might have had some edges of its own that have to be removed but by correcting the labels one by one, we can ensure that once we fix vertex $v$ by assigning it to label $i$ with vertex $u$ having label $i$ before, then later we will have to fix the label of vertex $u$, resolving the issue.)
%
%\Stefan{But that vertex might have had some edges of its own that have to be removed... I think we cannot really do label assignments we can only do label swaps. Swapping two labels can be done in $\le  2d$ steps, hence swapping labels is no more powerful than edge edits. Moreover, a sequence of $k$ label assignments that starts and ends with no repeated labels can be simulated with $\le  k$ label swaps.}
%\Artur{Not really. We correct the labels one by one, and so once we fix vertex $v$ by assigning it to label $i$ with vertex $u$ having label $i$ before, then later we will have to fix the label of vertex $u$. So at the end, all will work well.}
Observe that if we apply this operation to all vertices with the labels to be changed, then we do not increase the maximum degree, and hence, modifying of $k$ labels can be simulated by modifying at most $2dk$ edges. Therefore, if a labeling $\num$ is \eps-far from a valid DFS numbering according to the edges modifications then $\num$ is also $(2d\eps)$-far from a valid DFS numbering according to the edges and labels modifications.
\end{proof}

%-------------------------------------------------------------------------------------------------------

\subsubsection{Why should the labels be a permutation?}
\label{subsubsec:discussion-about-model-non-permutations}
Our model assumes that the input labeling $\num$ is a permutation (bijection) of $V$ to $\{1\,\dots,n\}$, but it may be natural to consider the case when $\num$ is \emph{not necessarily a permutation} but rather an arbitrary function from $V$ to $\{1\,\dots,n\}$, allowing repetitions of labels. Observe that already the simple problem of testing if $\num$ is a permutation or is \eps-far from being a permutation (also known as the element distinctness problem) is known to require $\Theta(\sqrt{n}/\eps)$ queries (see, e.g., \cite{RRSS09}).
%\Artur{Christian: Is \cite{RRSS09} the right reference to the complexity of testing element distinctness?}
In view of that bound, the best what we could hope for without assuming that $\num$ is a permutation of $V$ to $\{1\,\dots,n\}$ is to achieve query complexity of $\Theta(\sqrt{n}/\eps)$. And this can be easily obtained by combining our algorithm in \cref{thm:DFS-ub} with the known algorithms testing element distinctness, leading to a testing algorithm with $\Theta(\sqrt{n}/\eps)$ queries.

%-------------------------------------------------------------------------------------------------------

\section{Basic properties of DFS numberings}
\label{section:DFS}

We begin with a review of basic DFS terminology, introduce some notions of our own, and make a few useful observations. (We also refer to standard textbooks, e.g., \cite{CLRS22,KT13}.)

The DFS algorithm, as given in Algorithm~\ref{alg:DFS-numbering-def}, numbers every vertex even if $G$ is \emph{disconnected} due to the outer loop over all vertices.
Every vertex is initially \emph{undiscovered}. When \dfsVisit($v$) is called, $v$ becomes \emph{discovered}. We say $v$ is \emph{active} as long as the call \dfsVisit($v$) persists and is \emph{finished} as soon as it ends. The vertex $p$ that initiates the call of \dfsVisit($v$) is the \emph{parent} of $v$. If the call of \dfsVisit($v$) is initiated from the outer loop then we say that $v$ is an \emph{orphan} with \emph{virtual parent} $0$. The idea is that the outer loop iterating over all vertices is like a call to \dfsVisit($0$) where $0$ is a virtual vertex connected to all other vertices. Each connected component of $G$ has exactly one orphan, which receives the smallest number in its connected component. At any point during the execution the \emph{white path} consists of all active vertices including the virtual vertex $0$, with every active vertex (other than $0$) connected to its parent. The DFS numbers appear in increasing order along the white path.
In an implicitly labeled graph $G = (V, E)$% with $V = [n]$,
, we define $p(v)$ for $v \in V$ as
\begin{align*}
    p(v) :=
    \begin{cases}
        \max\{u \in N(v) \cap  [v-1]\} & \text{ if } N(v) \cap [v-1] \ne \emptyset \enspace,\\
        0 & \text{ otherwise.}
    \end{cases}
\end{align*}

\begin{claim}
\label{claim:dfs-parent-is-p}
If the numbers in an implicitly labeled graph correspond to a DFS numbering then $p(v)$ is the (possibly virtual) parent of $v$.
\end{claim}

\begin{proof}
If $v$ is an orphan then it has the smallest number in its connected component so $\{u \in N(v): u < v\} = \emptyset$. Hence $p(v) = 0$, which is the virtual parent of orphans by definition.

Now assume $v$ was discovered from a non-virtual parent $x \in  V$. Then $x \in  N(v)$ with $x < v$ so $x \in  \{u \in N(v): u < v\}$. To see that $x$ is the maximum of $\{u \in N(v): u < v\}$, consider $x' \in  \{x+1,\dots,v-1\}$. Immediately before $v$ is discovered, $x$ is at the end of the white path, hence the largest active vertex. Since $x' > x$ has been discovered, $x'$ must be finished. Hence all neighbours of $x'$ have been discovered. Since $v$ is not discovered we have $x \notin  N(v)$ so $x \notin  \{u \in N(v): u < v\}$.
\end{proof}

Observe that any edge $e=\{u,v\}$ with $u < v$ in an implicitly labeled graph corresponds to an \emph{interval} $[u,v]$ representing a range of elements with respect to the DFS numbering. The analysis of the interrelation between such intervals plays a central role in our analysis.

\begin{definition}
\label{def:conflicting -pair}
%We say a pair $(v,\{u,w\}) \in  V \times E$ is a \emph{conflicting pair} or, in short, \emph{conflict}, if $p(v) < u < v < w$.
A pair $(v,\{u,w\}) \in  V \times  E$ is a \textbf{conflicting pair} if $p(v) < u < v < w$.
\end{definition}

The following central lemma shows that the absence of conflicting pairs is necessary and sufficient for the validity of a DFS numbering.

\begin{lemma}
\label{lem:dfs-characterization}
Let $G = (V,E)$ be an implicitly labeled graph. The following are equivalent.
\begin{enumerate}[\rm (i)]
\item $G$ has a valid DFS numbering.
\item There exists no conflicting pair.
\end{enumerate}
\end{lemma}

\begin{proof}
\begin{description}
\item[(i) $\Rightarrow $ (ii).]
Consider a DFS run that agrees with the DFS numbering and let $(v,\{u,w\}) \in  V \subseteq  E$ be a pair with $p(v) < u < v$. By \cref{claim:dfs-parent-is-p}, $v$ was discovered from $p(v)$. Consider the time immediately before this happens. Vertex $u$ was already discovered (due to $u < v$) and is not active (since $p(v) < u$ is the end of the white path and $p(v)$ is the largest active vertex), meaning $u$ is finished. %\Artur{Another way of completing the proof now would be to say that then edge $\{u,w\}$ with $u < u < w$ is a cross edge, and hence it cannot exist in undirected graphs, see, e.g., \cite[Theorem~20.10]{CLRS22}.}
Thus all neighbors of $u$ are already discovered, so $w < v$. In particular $(v,\{u,w\})$ is not a conflicting pair.
\item[(ii) $\Rightarrow $ (i).]
We use induction on $k$, showing that there is a DFS run assigning number $v$ to each $v \in  \{1,\dots,k\}$. For $k = 0$ the claim is trivial. Assume the claim holds for some $k-1$. Consider a corresponding DFS run that has just discovered vertex $k-1$; at that moment vertex $k$ is undiscovered. We have $0 \le  p(k) \le  k-1$. If $p(k)$ does not appear on the white path then, because both $0$ and $k-1$ appear on the white path, there is a vertex $x$ on the white path with $p(x) < p(k) < x < k$. This makes $(x,\{p(k),k\})$ a conflicting pair, contradicting (ii). Therefore $p(k)$ does appear on the white path. There might be a vertex $u > p(k)$ further along on the white path. Let $w$ be a neighbor of $u$. Due to $p(k) < u < k$ we must have $w \le  k$, for otherwise $(k,\{u,w\})$ would be a conflicting pair. Moreover we have $w \neq  k$ by definition of $p(k)$. Thus $w \le  k-1$, meaning $w$ was already discovered. In particular the DFS will backtrack from active vertices until it reaches $p(k)$, from which it may legally choose to discover $k$ next. This concludes the induction.\qedhere
\end{description}
\end{proof}

\junk{

\begin{proof}
\textcolor[rgb]{0.54,0.17,0.89}{Proof for directed graphs.}
\begin{description}
\item[(i) $\rightarrow$ (ii).] Consider a DFS run that agrees with the DFS numbering and let $(v,(u,w)) \in V \times E$ be a pair with $p(v) < u < v$. By \cref{claim:dfs-parent-is-p}\Artur{Which wold have to written (in the same way) for directed graphs.}, $v$ was discovered from $p(v)$. Consider the time immediately before this happens. Vertex $u$ was already discovered (due to $u < v$) and is not active (since $p(v) < u$ is the end of the white path and $p(v)$ is the largest active vertex), meaning $u$ is finished. Therefore, since the finishing time of vertex $u$ is between $u$ and $w$, by a known property of DFS (see, e.g., \cite[Excercise~20.3-5]{CLRS22}), we cannot have edge $(u,w) \in E$ with $p(v) < u < v < w$. In particular $(v,(u,w))$ is not a conflicting pair.
\item [(ii) $\rightarrow$ (i).] We use induction on $k$, showing that there is a DFS run assigning number $v$ to each $v$, $1 \le v \le k$. For $k = 0$ the claim is trivial. Assume the claim holds for some $k-1$. Consider a corresponding DFS run that has just discovered vertex $k-1$; at that moment vertex $k$ is undiscovered. We have $0 \le  p(k) \le  k-1$. If $p(k)$ does not appear on the white path then we infer that because both $0$ and $k-1$ appear on the white path, there is a vertex $x$ on the white path with $p(x) < p(k) < x < k$. This makes $(x,(p(k),k))$ a conflicting pair, contradicting the assumption (ii). Therefore $p(k)$ does appear on the white path. There might be a vertex $u > p(k)$ further along on the white path. Let $w$ be a neighbor of $u$. Due to $p(k) < u < k$ we must have $w \le  k$, for otherwise $(k,(u,w))$ would be a conflicting pair. Moreover we have $w \neq  k$ by definition of $p(k)$. Thus $w \le  k-1$, meaning $w$ was already discovered. In particular the DFS will backtrack from active vertices until it reaches $p(k)$, from which it may legally choose to discover $k$ next. This concludes the induction.\qedhere
\end{description}
\end{proof}

}

%-------------------------------------------------------------------------------------------------------

\section{Testing DFS numbering requires \texorpdfstring{$\Omega(n^{1/3})$}{Omega(n**(1/3))} Queries}
\label{section:DFS-lb}

%Our first main result shows that testing DFS numbering requires $\Omega(n^{1/3})$ queries.
%In this section, we prove our first main result that any algorithm testing DFS numbering requires $\Omega(n^{1/3})$ queries.
In this section, we prove our first main result.

\begin{theorem}
\label{thm:DFS-lb}
Every property tester for the property of having a valid DFS-labeling has a query complexity of $\Omega(n^{1/3})$.
\end{theorem}

Let $\eps > 0$ be sufficiently small (any $\eps \le \frac{1}{33}$ would do). The proof of \cref{thm:DFS-lb} is by constructing two families $(G_n)_{n \in \NN}$ and $(B_n)_{n \in \NN}$ of \emph{good} and \emph{bad} random labeled graphs for which we will show that distinguishing between these families is necessary for any DFS-tester. Then we will show that distinguishing between these families requires $\Omega(n^{1/3})$ queries.

%-------------------------------------------------------------------------------------------------------

\subsection{Construction of good and bad random labeled graphs
\texorpdfstring{$(G_n)$ and $(B_n)$}{(Gn) and (Bn)}
}
\label{subsection:DFS-lb-construction}

Let $n, N \in \NN$. Each graph $G_n$ and $B_n$ consists of $\lfloor\frac{n}{8N}\rfloor$ \emph{arms} of $8N$ vertices each. The roots of these arms are connected with some binary tree that is the same for $G_n$ and $B_n$.
\begin{figure}[t]
\centerline{\includegraphics[width=.475\textwidth]{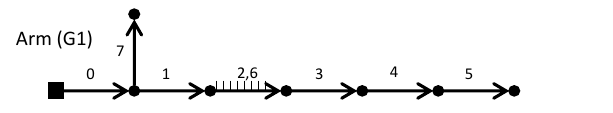} \quad \includegraphics[width=.475\textwidth]{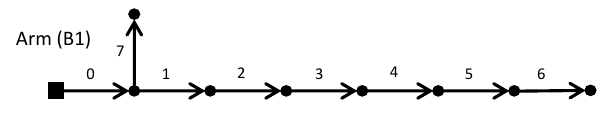}}
\centerline{\includegraphics[width=.475\textwidth]{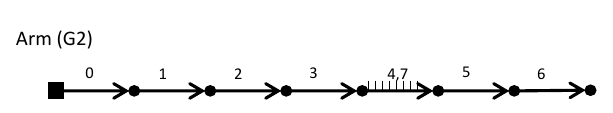} \quad \includegraphics[width=.475\textwidth]{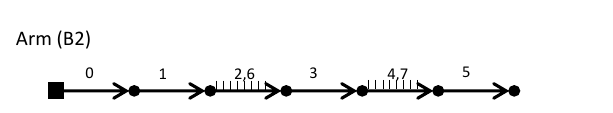}}
\caption{Four types of \emph{arms} $(G1)$, $(G2)$, $(B1)$, and $(B2)$ used in our construction of $(G_n)_{n \in \NN}$ and $(B_n)_{n \in \NN}$. Each arm starts with a root, denoted by {\footnotesize \textcolor[rgb]{0.00,0.00,0.00}{$\blacksquare$}}. Each link \includegraphics[width=.075\textwidth]{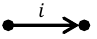} corresponds to a path on $N$ vertices and labels $iN+1, \dots, i(N+1)$ ascending in the direction of the arrow (see Figure~\ref{fig-lb-1a}(a)). We also have a \emph{comb graph} \includegraphics[width=.075\textwidth]{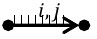} (see Figure~\ref{fig-lb-1a}(b)), which is obtained from \includegraphics[width=.075\textwidth]{DFS-LB-f7} by adding $N$ new vertices, adding a matching between the new vertices and the vertices from \includegraphics[width=.075\textwidth]{DFS-LB-f7}, and labeling the new vertices with $jN+1, \dots, (j+1)N$, this time descending in the direction of the arrow (i.e., the vertex with label $iN+k$ is adjacent to the vertex with label $jN+N-k+1$).}
\label{fig-lb-1}
\end{figure}

There are four types of arms, as described in details in Figure \ref{fig-lb-1}.
\begin{figure}[t]
\centerline{(a) \includegraphics[width=.45\textwidth]{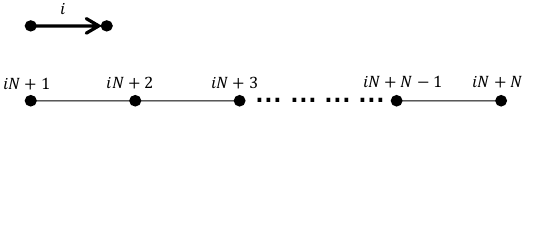} \quad (b) \includegraphics[width=.45\textwidth]{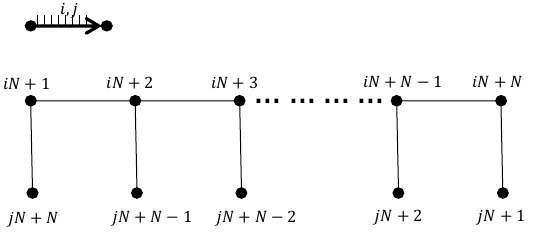}}
\caption{Graphs corresponding to (a) a path \includegraphics[width=.075\textwidth]{DFS-LB-f7} and (b) a comb graph \includegraphics[width=.075\textwidth]{DFS-LB-f8}% with parameters $i,j$
.}
\label{fig-lb-1a}
\end{figure}
Each arm of $G_n$ is a copy of $(G1)$ or $(G2)$ chosen independently and uniformly at random. Similarly, each arm of $B_n$ is a copy of $(B1)$ or $(B2)$ chosen independently and uniformly at random.

The labeling of $G_n$ is then obtained by starting in the root of the tree $T$ and using within the arms the relative ordering defined by the numbering of the arms $(G1)$ or $(G2)$. For an example, see, e.g., Figure~\ref{fig-lb-2a}.
(Since each arm has the same number of vertices, the choice of one arm, whether in $(G1)$ or $(G2)$, does not affect the numbering of other arms.)

Similarly (see Figure~\ref{fig-lb-2b}), the labeling of $B_n$ is defined by starting at a root of $T$ and using within the arms the relative ordering defined by the numbering of the arms $(B1)$ or $(B2)$. 

%-------------------------------------------------------------------------------------------------------

\begin{figure}[th]
\centerline{\includegraphics[width=\textwidth]{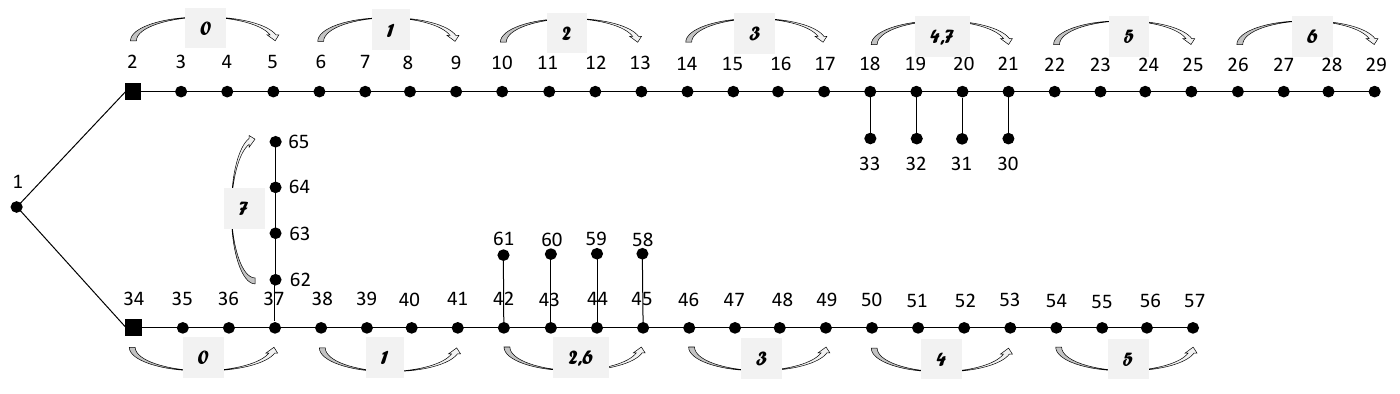}}
\caption{An example of labeling of $G_n$ with $n=64$, $N = 4$, using $\lfloor\frac{n}{8N}\rfloor = 2$ arms of size $8 \cdot 4 = 32$. Vertices with labels 1, 2, 34 are in the tree $T$, the top branch (with labels 2--33) corresponds to arm $(G2)$ and the bottom branch (with labels 34--65) corresponds to arm $(G1)$. For each arm, we also marked the corresponding parts defining it.}
\label{fig-lb-2a}
\end{figure}

\begin{figure}[th]
\centerline{\includegraphics[width=\textwidth]{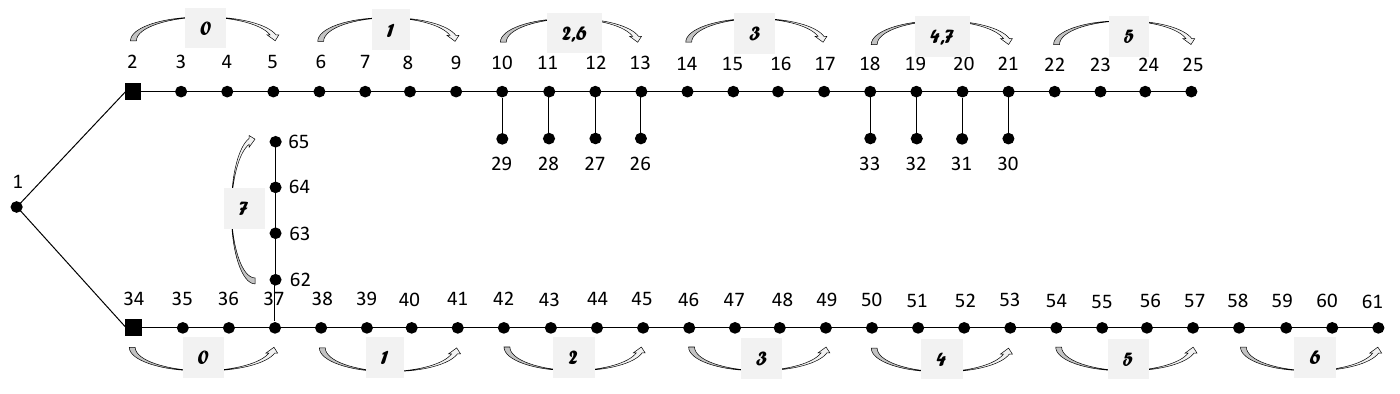}}
\caption{An example of labeling of $B_n$ with $n=64$, $N = 4$, using $\lfloor\frac{n}{8N}\rfloor = 2$ arms of size $8 \cdot 4 = 32$. Vertices with labels 1, 2, 34 are in the tree $T$, the top branch (with labels 2--33) corresponds to arm $(B2)$ and the bottom branch (with labels 34--65) corresponds to arm $(B1)$. For each arm, we also marked the corresponding parts defining it.
}
\label{fig-lb-2b}
\end{figure}

%-------------------------------------------------------------------------------------------------------

\subsection{Properties of DFS numberings of random labeled graphs
\texorpdfstring{$G_n$ and $B_n$}{(Gn) and (Bn)}
}
\label{subsection:DFS-lb-properties}

Let us first notice that our construction of good random labeled graphs $(G_n)_{n \in \NN}$ easily ensures that they have valid DFS numberings (with probability 1).

\begin{lemma}
\label{DFS-lb-prop1}
$G_n$ has a valid DFS numbering.
\qed
\end{lemma}

\junk{
\begin{proof}
This easily follows by observing that each arm $(G1)$ or $(G2)$ maintains locally a valid DFS-order (see, e.g., \cref{fig-lb-2a} in \cref{sec:aux-figures-lb}).
\end{proof}
}

A situation is different for bad random labeled graphs $(B_n)_{n \in \NN}$. While the arms of type $(B1)$ also locally maintain a valid DFS-order, the arms of type $(B2)$ do not (because of the two comb graphs). As the result, the construction of $(B_n)_{n \in \NN}$ typically gives labelings that are \eps-far from valid DFS numberings.

\begin{lemma}
\label{DFS-lb-prop2}
Let
%\Artur{You could also ignore \eps and just write that $B_n$ has a labeling that is $\frac{1}{33}$-far from \dots}
$0 \le \eps \le \frac{1}{33}$ and let $n$ be sufficiently large. Then $B_n$ has a labeling that is \eps-far from a valid DFS numbering with probability $1 - o(1)$.
\end{lemma}

\begin{proof}
Consider the numbering for $(B2)$. Let us fix any $k \in \{1,2,\dots,N\}$ and focus on vertices $p_1, c_1, p_2, c_2$ with labels $2N+k, 7N-k+1, 4N+k, 8N-k+1$, respectively. (For example, in Figure~\ref{fig-lb-2b}, in the top branch (where numbers have an offset of $1$), if we took $k=2$ then we would have $\langle p_1, c_1, p_2, c_2 \rangle$ to be vertices with labels $\langle 11, 28, 19, 32 \rangle$ in $B_{64}$.) Observe that the fact that $\num(p_1) < \num(p_2) < \num(c_1) < \num(c_2)$ implies that this is not a valid DFS numbering as long as $c_1$ is the DFS-child of $p_1$ and $c_2$ is the DFS-child of $p_2$. Therefore, to turn the labeled graph into one consistent with DFS numbering, one has to add or delete at least one edge incident to a vertex from $\{p_1,c_1,p_2,c_2\}$. Since such a quadruple is obtained for each $k \in \{1,\dots,N\}$ and since each of these quadruples is disjoint (for example, in Figure~\ref{fig-lb-2b}, the there are four such quadruples formed by vertices with labels $\langle 10, 29, 18, 33 \rangle$, $\langle 11, 28, 19, 32 \rangle$, $\langle 12, 27, 20, 31 \rangle$, and $\langle 13, 26, 21, 30 \rangle$), one needs to modify at least $\frac{N}{2}$ edges to obtain a valid DFS numbering. Since the expected number of copies of $(B2)$ in $B_n$ is $\frac12 \cdot \lfloor\frac{n}{8N}\rfloor$, the expected number of changes required in $B_n$ to obtain a valid DFS numbering is at least $\frac{N}{4} \cdot \lfloor\frac{n}{8N}\rfloor$. For sufficiently large $n$, a standard Chernoff bound implies that graph $B_n$ is $\frac{1}{33}$-far from having a valid DFS numbering with high probability.
\end{proof}

%-------------------------------------------------------------------------------------------------------

\subsection{Hardness of distinguishing between good and bad labeled graphs}
\label{subsection:key-fact-about-distinguishing}

Our next central lemma provides a lower bound for the number of queries of any algorithm \ALG that distinguishes between the families of random labeled graphs $(G_n)_{n \in \NN}$ and $(B_n)_{n \in \NN}$. We will assume that the input $I_n$ for $n \in \NN$ on which \ALG requires $q_n$ oracle queries is obtained by first selecting a random bit $b \in \{0,1\}$ and then setting %$I_n = G_n$ if $b = 0$ and $I_n = B_n$ if $b = 1$.
$I_n = \begin{cases}G_n & \text{ if } b = 0, \\ B_n & \text{ if } b = 1.\end{cases}$
%We will formalize our result in the following lemma.

\begin{lemma}
\label{DFS-lb-prop3}
Let $0 < \xi \le \frac{N^3}{n}$. For every randomized algorithm \ALGI that receives $I_n$ as input, performs $q_n \le \sqrt{\frac{\xi n}{4N}}$ queries, and outputs a bit $b' \in \{0,1\}$, we have $\Pr[b = b'] \le \frac12 + \xi$.
\end{lemma}

\begin{proof}
Let \ALGI be an algorithm that receives $I_n$ as its input. We assume that \ALGI can query the oracle for any vertex $u$ by submitting an ID of $u$, and the oracle returns the label $\num(u)$ of $u$, and the IDs of all neighboring vertices. Further, without loss of generality, we assume that the IDs $1, \dots, n$ are randomly assigned to the vertices of the graph. Therefore \ALGI is limited to querying for a random vertex (a \emph{random query}) or for a vertex that has been previously found as a neighbor of an earlier queried vertex (an \emph{explorative query}).

Without loss of generality, we can assume that the vertices selected by random queries are chosen in the order $v_1, v_2, \dots$ before the run of \ALGI; let $V_R = \{v_1, \dots, v_{q_n}\}$ be the sequence of these first $q_n$ random vertices (\ALG can select only these vertices). Let \EPS be the random event that no two vertices $u, u' \in V_R$ come from the same arm of $I_n$. We prove the following:
\begin{align}
    \Pr[\neg\EPS] \le \xi
        \label{DFS-lb-prop3-1} \\
    \Pr[b = b' | \EPS] = \tfrac12
        \label{DFS-lb-prop3-2}
\end{align}

Observe that the two inequalities (\ref{DFS-lb-prop3-1})--(\ref{DFS-lb-prop3-2}) imply \cref{DFS-lb-prop3}:
\begin{align*}
    \Pr[b = b'] &=
        \Pr[b = b' | \EPS] \cdot \Pr[\EPS] + \Pr[b = b' | \neg\EPS] \cdot \Pr[\neg\EPS]
    \le
        \tfrac12 \cdot 1 + 1 \cdot \xi
    \le
        \tfrac12 + \xi
    \enspace.
\end{align*}

In order to prove inequality (\ref{DFS-lb-prop3-1}), let us first define $C$ to be the number of pairs $i \ne j$ with $i, j \le q_n$ such that both $v_i$ and $v_j$ come from the same arm of $I_n$, or in other words, so that $v_i$ and $v_j$ belong to the same copy of arm $(B1)$. We have,
\begin{align*}
    \Ex{C} &=
    \!\!\!\!\!\sum_{1 \le i < j \le q_n} \!\!\!\!\!\Pr[\text{$v_i$ and $v_j$ are in the same arm}]
        \le
    \binom{q_n}{2} \cdot \frac{8N}{n}
        \le
    \frac{4 q_n^2 N}{n}
        \le
    \frac{4 (\frac{\xi n}{4N}) N}{n}
        \le
    \xi
    \enspace.
\end{align*}
Therefore we can conclude (\ref{DFS-lb-prop3-1}) by Markov's inequality: $\Pr[\neg\EPS] = \Pr[C \ge 1] \le \Ex{C} \le \xi$.

Next, we want to prove inequality (\ref{DFS-lb-prop3-2}). Let $V_R^*$  be the set of vertices reachable from vertices in $V_R$ in at most $q_n$ steps, that is, $V_R^* = \{ u: \exists_{1 \le i \le q_n} \dist(v_i,u) \le q_n\}$. Let $I_n(V_R^*)$ be the subgraph of $I_n$ induced by the vertex set $V_R^*$. We observe that \ALGI will make its decision solely on seeing some subgraph of $I_n(V_R^*)$. Hence, the output $b'$ of \ALGI is a (random) function of $I_n(V_R^*)$. Now, we claim that conditioned on \EPS, the random variables $b$ and $I_n(V_R^*)$ are stochastically independent, which in turn, would imply identity (\ref{DFS-lb-prop3-2}).

To prove that $b$ and $I_n(V_R^*)$ are independent, let us consider a single vertex $v_i$ from $V_R$ and the subgraph $I_n(v_i)$ induced by vertices with distance at most $q_n$ from $v_i$. If $I_n(v_i)$ contains at least one vertex from $T$, then $I_n(v_i)$ consists solely of some of the vertices from $T$ and some vertices that are close to the roots of some of the arms (within parts denoted by $\xrightarrow{0}$ of the arms, since each such part has length $N$ and we have $q_n \le \frac{N}{2}$ since $q_n \le \sqrt{\frac{\xi n}{4N}}$ and $\xi \le \frac{N^3}{n}$). Since these parts are identical in $G_n$ and $B_n$, such paths share no information on $b$.

Otherwise, if $I_n(v_i)$ contains no vertex from $T$, then $I_n(v_i)$ is a part of an arm of $I_n$. But then, due to $q_n \le \frac{N}{2}$, at most one joint\footnote{By a \emph{joint} we mean an endpoint of a path or comb from which the arm was stitched together.} of this arm is contained in $I_n(v_i)$. Since the labels of the roots of the arms are the same in $G_n$ and $B_n$, it is always well-defined what the offset of a label within its arm is.
%\Artur{This part requires polishing; now all is deferred to \cref{fig-lb-3} and the text below it.}
For more details, see \cref{sec:proof:final-DFS-lb-prop3} and Figure~\ref{fig-lb-3} therein.
\end{proof}

%-------------------------------------------------------------------------------------------------------

\subsection{Hardness of testing DFS numbering (proof of Theorem \ref{thm:DFS-lb})}
\label{subsection:DFS-lb-hardness}

By \cref{DFS-lb-prop1,DFS-lb-prop2}, any algorithm \ALG that accepts a labeled graph with a valid DFS numbering with probability at least $\frac23$ and rejects a labeled graph with a DFS numbering that is \eps-far (for $0 < \eps \le \frac1{33}$) from being valid DFS numbering with probability at least $\frac23$, must be able to distinguish with probability at least $\frac58$ between the families of good and bad random labeled graphs $(G_n)_{n \in \NN}$ and $(B_n)_{n \in \NN}$. However by \cref{DFS-lb-prop3}, by setting $\xi = \frac18$ and $N = \lfloor n^{1/3} \rfloor$, the tasks of distinguishing between these families requires $\Omega(n^{1/3})$ queries.
%
%It is not difficult to see that Lemmas \ref{DFS-lb-prop1}--\ref{DFS-lb-prop3} directly imply \cref{thm:DFS-lb} because any tester for DFS numbering is a successful distinguisher of $(G_n)_{n \in \NN}$ and $(B_n)_{n \in \NN}$, and therefore by \cref{DFS-lb-prop3}, it requires $\Omega(\frac{1}{k^2})$ queries.\Artur{What is $k$??? Why are these arguments giving a proof of \cref{thm:DFS-lb}?}
%
\qed

%\includepdf[pages=-]{DFS-LowerBound-fixed.pdf}

%-------------------------------------------------------------------------------------------------------

%\newpage
\section{Testing DFS numbering with
\texorpdfstring{$O(n^{1/3}/\eps)$}{O(n**(1/3)/epsilon}
queries}
\label{section:DFS-ub}

Our second main result shows that the lower bound in \cref{thm:DFS-lb} is asymptotically tight. % by designing an algorithm testing DFS numbering with $O(n^{1/3}/\eps)$ queries.
%We present the proof for directed graphs, keeping in mind that the standard reduction implies a similar result also for undirected graphs.\Artur{Or something similar should be said here.}

\begin{theorem}
\label{thm:DFS-ub}
Let $0 < \eps < 1$. There is an algorithm that with oracle access to a labeled bounded degree undirected graph $G$ on $n$ vertices,
%\Stefan{We actually prove: „with running time $O(n^{1/3}/\eps )$. Should we state it as such (also in section heading)?}
performs $O(n^{1/3}/\eps)$ queries to the oracle and
%\Stefan{Actually, it always accepts valid DFS. State it as such?}
accepts, if $G$ has a valid DFS numbering, and rejects with probability at least $\frac23$, if $G$ is \eps-far from having a valid DFS numbering.
\end{theorem}

Our tester in \cref{thm:DFS-ub} has one-sided error and always accepts a valid DFS numbering. It achieves not only $O(n^{1/3}/\eps)$ query complexity, but also it has $O(n^{1/3}/\eps)$ running time, see \cref{sec:sweep-line-decider}.

\medskip

For the rest of the section, let $G$ be a corresponding labeled graph equipped with a possibly inappropriate DFS numbering $\num : V \rightarrow [n]$. As in \cref{section:DFS}, we assume that \emph{$\num$ is implicit} so that we can, for instance, write $v < w$ for $v,w \in  V$ instead of $\num(v) < \num(w)$.

%-------------------------------------------------------------------------------------------------------

\paragraph{Outline.}
%\subsection{Outline}
\label{subsec:outline-ub}
Our approach to prove \cref{thm:DFS-ub} is first to extend \cref{lem:dfs-characterization} (which characterizes valid DFS numberings) to describe a simple and useful property of labelings that are \eps-far from a valid DFS numbering. In \cref{lem:dfs-characterization-epsilon-far} in \cref{subsec:properties-eps-far}, we will show that if the numbering is \eps-far from a valid DFS numbering then not only we have $\Omega(\eps n)$ conflicting pairs (in the sense of \cref{lem:dfs-characterization}), but in fact we have $\Omega(\eps n)$ conflicting pairs that are ``unrelated.''  Once we have this property, the task in hand will be to detect any of such conflicting pair. We observe that there are two types of conflicting pairs, \emph{local pairs} involving vertices whose DFS numbers are close to each other, and \emph{global conflicts}. In order to detect local conflicts, we first develop (in \cref{subsec:Navigating}) some basic tools to traverse a given graph following DFS numbering. Once we know how to traverse the graph, we can design an algorithm that can determine if a given pair $(v,\{u,w\})$ is conflicting pair. Unfortunately, this algorithm is efficient only if vertex $u$ or $w$ is close to $p(v)$ or $v$ (that is, if one of $\num(u) - \num(p(v))$, $\num(v) - \num(u)$, $\num(w) - \num(v)$ is small), and so we can use this approach to deal with local conflicts. In order to study global conflicts, we notice that if vertices $u$ and $w$ are far away from vertices $p(v)$ and $v$, then in fact a conflicting pair $(v,\{u,w\})$ can be also extended to multiple vertices $v$. Once we have that property, we will show that if we sample randomly $\Theta(n^{1/3}/\eps)$ vertices and $\Theta(n^{1/3}/\eps)$ edges, then if there were many global conflicts, then there would be one that is determined by one of the sampled vertices and one of the sampled edges.

%-------------------------------------------------------------------------------------------------------

\subsection{Properties of labelings that are \texorpdfstring{\eps}{epsilon}-far from any valid DFS numbering}
\label{subsec:properties-eps-far}

We begin with describing a useful property of labelings that are \eps-far from a valid DFS numbering that they have $\Omega(\eps n)$ conflicting pairs that are ``unrelated.''
Recall that by \cref{lem:dfs-characterization} a numbering is a valid DFS numbering if and only if there is no conflicting pair $(v,\{u,w\}) \in  V \times  E$ %\Artur{You wrote conflicting pair $(v,\{u < w\}) \in  V \times  E$; I think there is no reason to mention $u<w$ here. However, in a few places below one should keep it (though then I would prefer ``$\{u,w\}$ with $u<w$'').}
with $p(v) < u < v < w$. In that case, when $p(v) < u < v < w$, we speak of a conflict involving vertex $v$ and edge $\{u,w\}$. While \cref{lem:dfs-characterization} characterizes valid and invalid DFS numberings, we will need a stronger claim about properties of numberings that are \eps-far from valid DFS numberings. For that, we need to understand numberings that are not \eps-far from valid DFS numberings because we can modify the input graph with at most $\eps n$ edge modifications to ensure that the resulting graph will have a valid DFS numbering.

Observe that \cref{lem:dfs-characterization} provides a simple tool to edit the input labeled graph to obtain a valid DFS numbering --- to remove all conflicting pairs. With that in mind, the following claim provides a simple fix to remove all conflicts involving a specific vertex or all conflicts involving a specific edge (notice that the resulting graph may violate our degree bound $d$, but one can address this issue using a framework developed for that in \cref{section:deg-reduction-DFS}).

\begin{lemma}
\label{claim:how-to-fix}
Let $v \in  V$ and $\{u,w\} \in  E$ with $u < w$.
%\Artur{Observe that \cref{claim:how-to-fix} works for directed graphs as well.}
\begin{enumerate}[(i)]
\item Adding the edge $\{v-1,v\}$ to $G$ (and doing nothing if $\{v-1,v\}$ is already present) does not create any new conflicts and resolves all conflicts involving $v$.
\item Removing $\{u,w\}$ from $G$ and adding $\{w-1,w\}$ (if not already present) does not create any new conflicts and resolves all conflicts involving $\{u,w\}$.
\end{enumerate}
\end{lemma}

\begin{proof}
\begin{enumerate}[(i)]
\item After the edit we have $p(v) = v-1$ so no $u \in  V$ can satisfy $p(v) < u < v$ any more, meaning all conflicts involving $v$ are resolved. We do not create any new conflicts because $p(v-1)$ does not change and the new edge $\{v-1,v\}$ cannot be involved in a conflict because (again) no $v'$ can satisfy $v-1 < v' < v$.
\item Deleting $\{u,w\}$ clearly resolves all conflicts of this edge. However, new conflicts involving $w$ may arise since $p(w)$ may change if we had $p(w) = u$ before. By (i) we can fix these by adding $\{w-1,w\}$.\qedhere
\end{enumerate}
\end{proof}

\def\C{\mathcal{C}}
\cref{claim:how-to-fix} shows that adding or removing a few edges can resolve many related conflicts. We use this claim to show that in order for $G$ to be $\eps $-far from having a valid DFS numbering, not only do there have to be many conflicts, but in fact there have to be \emph{many mutually unrelated conflicts}. To formalize this idea we consider the bipartite \emph{conflict graph} $\C = (V,E,C)$ where $C \subseteq  V \times  E$ contains an edge $(v,\{u, w\})$ precisely if it is a conflicting pair. The intuition of mutually unrelated conflicts corresponds to a matching in $\C$. We have the following.

\begin{lemma}[\textbf{\eps-far DFS numberings}]
\label{lem:dfs-characterization-epsilon-far}
If $G$ is $\eps $-far from having a valid DFS numbering then there is a matching $M \subseteq  C$ of size $|M| \ge  \eps n/5$ in $\C$.
%\Artur{Observe that \cref{lem:dfs-characterization-epsilon-far} works for directed graphs as well.}
\end{lemma}

\begin{proof}
\def\Cov{\mathcal{VC}}
We will prove \cref{lem:dfs-characterization-epsilon-far} by showing that if $G$ is \eps-far from having a valid DFS numbering then $\C$ has a vertex cover of size at least $\eps n/5$, for if not, then we could modify at most $\eps n$ edges of $G$ to make the labeling to be a valid DFS numbering of the resulting graph. Then  \cref{lem:dfs-characterization-epsilon-far} follows from K\"{o}nig's theorem.

Let $M \subseteq  C$ be a maximum matching in $\C$. By Kőnig's theorem there is a vertex cover $\Cov \subseteq  V \cup  E$ of size $|\Cov| = |M|$, i.e., a set of vertices and edges of $G$ (vertices of $\C$) such that for every conflict pair $(v,\{u,w\})$ (for every edge of $\C$) $v \in  \Cov$ or $\{u,w\} \in  \Cov$. We can fix all conflicts in $G$ in $\le  2 |\Cov|$ edits by applying for each $v \in  \Cov \cap  V$ the fix of \cref{claim:how-to-fix}~(i) (one edit) and for each $\{u,w\} \in  \Cov \cap  E$ the fix of \cref{claim:how-to-fix}~(ii) (two edits). We obtain a graph $G^*$ with a valid DFS numbering. However, the vertices that appeared in a fix in the role of $v$ or $w$ may have degree $d+1$ in $G^*$, higher than permitted. By 
our ``degree reduction framework'' from \cref{lemma:deg-reduction-DFS} (see \cref{section:deg-reduction-DFS}),
    %\Artur{We should mention something about this lemma or section before we first time have to deal with the problem of increasing the degrees.}
another $3|\Cov|$ edits suffices to transform $G^*$ into $G^{**} = (V,E^{**})$ with maximum degree $d$ while maintaining the validity of the DFS numbering. Overall $|E^{**} \triangle  E| \le  5|\Cov|$ and since $G$ is $\eps $-far we have $|E^{**} \triangle  E| \ge  \eps n$. Hence $|M| = |\Cov| \ge  \eps n/5$.
\end{proof}

\cref{lem:dfs-characterization-epsilon-far} immediately implies a simple tester detecting \eps-far instances: we randomly sample $\Omega(\sqrt{n/\eps})$ vertices and $\Omega(\sqrt{n/\eps})$ edges, and then %(by birthday-paradox like arguments)
with a constant probability one of the sampled vertices $v$ and one of the sampled edges $e$ will form a conflicting pair $(v,e)$. %This yields the following corollary of \cref{lem:dfs-characterization-epsilon-far}.

%\cref{lem:dfs-characterization-epsilon-far} implies that an \eps-far instance has $\Omega(\eps n)$ pairwise distinct vertices and edges involved in conflicts. This immediately implies a simple tester detecting \eps-far instances: we randomly sample $\Omega(\sqrt{n/\eps})$ vertices and $\Omega(\sqrt{n/\eps})$ edges, and then (by birthday-paradox like arguments) with a constant probability one of the sampled vertices $v$ and one of the sampled edges $e$ will form a conflicting pair $(v,e)$. This yields the following corollary of \cref{lem:dfs-characterization-epsilon-far}.
%
%\footnote{Let us sample (i.u.r.) $s$ vertices from $V$. Observe that, with high probability, we have sampled at least $\eps s / 2$ vertices corresponding to some matching edges in $M$. Let $\mathcal{E}$ be the set of edges matched in $M$ to the sampled vertices. If we sample (i.u.r.) $s$ edges from $E$, then the probability that none of these edges is in $\mathcal{E}$ is at most
%%(it would have been exactly $\left(1 - \frac{|\mathcal{E}|}{|E|}\right)^s$ if we sampled with replacement, but without replacement it is only an upper bound)
% $\left(1 - \frac{|\mathcal{E}|}{|E|}\right)^s \le e^{-\frac{|\mathcal{E}| \cdot s}{|E|}} \le e^{-\frac{\eps s^2}{2|E|}}$. Therefore, if we sample $s$ vertices and $s$ edges with $s \ge c \cdot \sqrt{dn/\eps} \ge c \cdot \sqrt{|E|/\eps}$ for a sufficiently large constant $c$, then with probability at least 0.99 the sampling will spot at least one conflicting pair.}

\begin{corollary}
\label{corollary:DFS-ub-weak}
%Let $0 < \eps < 1$.
There is an algorithm that with oracle access to a labeled bounded degree~graph $G$ on $n$ vertices, performs $O(\sqrt{n/\eps})$ queries to the oracle and accepts, if $G$ has a valid DFS numbering, and rejects with probability %at least
$\frac23$, if $G$ is \eps-far from having a valid DFS numbering.
\qed
\end{corollary}

In what follows, we will show how to improve this result%
%and design a property tester with the query complexity $O(n^{1/3}/\eps)$
, as promised in \cref{thm:DFS-ub}. %Observe that only that bound matches the query complexity of the lower bound from \cref{thm:DFS-lb}.

%------------------------------------------------------------------------

\subsection{Navigating (would-be) DFS trees}
\label{subsec:Navigating}

%\cref{lem:dfs-characterization-epsilon-far} shows that if $G$ is \eps-far from %having
%a valid DFS numbering then there are $\Omega(\eps n)$ pairwise-unrelated conflicting pairs. Now, we develop tools to find conflicting pairs for vertices. We first show how to traverse a graph using consecutive labels (see \cref{sec:proof:fact:navigate-ordered-trees} for a proof).

In this section, we develop tools to find conflicting pairs for vertices. We first show how to traverse a graph using consecutive labels.

\begin{fact}
\label{fact:navigate-ordered-trees}
%\Artur{I find the formulation in this lemma rather convoluted. What is an ordered tree? What is canonical DFS ordering? (is this defined in the footnote below? it's not clear)}%
Let $T$ be an ordered tree of bounded degree and $S$ its canonical DFS ordering\footnote{The DFS starts at the root and respects the order of children when visiting them in pre-order traversal.}%
%\Artur{Why canonical? It might be OK, but isn't it more often called a \emph{pre-order tree traversal}?}
. %Given any vertex $v$ of $T$ and oracle access to $T$ it is possible to locate the successor (and the predecessor) of $v$ in $S$, if one exists, with $O(1)$ queries to $T$ on average overall vertices. %We can likewise locate the predecessor of $v$ in $S$.
For a randomly selected vertex $v$ from $T$, given oracle access to $T$, it is possible to locate the successor (and the predecessor) of $v$ in $S$, if one exists, with $O(1)$ queries to $T$ in expectation.
\end{fact}

%-------------------------------------------------------------------------------------------------------

\begin{figure}[t]
    \SetKwFor{Repeat}{repeat}{}{EndLoop}
    \begin{minipage}{0.48\linewidth}
        \begin{algorithm}[H]
            \caption{\textsc{tree-next}($v$)}
            \label{alg:tree-next}
            \If{$v$ has children $c_1,\dots ,c_k$}{
                \Return $c_1$\;
            }
            \Repeat{}{
                \If{$v$ has no parent}{
                    \Return None\;
                }
                \If{$v$ has a next sibling $s$}{
                    \Return $s$\;
                }
                $v \leftarrow \text{parent of $v$}$\;
            }
        \end{algorithm}
    \end{minipage}
    \begin{minipage}{0.48\linewidth}
        \begin{algorithm}[H]
            \caption{\textsc{tree-prev}($v$)}
            \label{alg:tree-prev}
            \If{$v$ has no parent}{
                \Return None\;
            }
            \If{$v$ has no previous sibling}{
                \Return parent of $v$\;
            }
            $w \leftarrow \text{previous sibling of $v$}$\;
            \While{$w$ has a child}{
                $w \leftarrow \text{last child of $w$}$\;
            }
            \Return $w$\;
        \end{algorithm}
    \end{minipage}
    \caption{Locate the successor and predecessor of a vertex $v$ in the DFS order of an ordered tree.}
    \label{algs:tree-next:tree-prev}
\end{figure}

%-------------------------------------------------------------------------------------------------------

\begin{proof}
We give here a rather straightforward proof of \cref{fact:navigate-ordered-trees} using traversal algorithms \textsc{tree-next} (Algorithm~\ref{alg:tree-next}) and \textsc{tree-prev} (Algorithm~\ref{alg:tree-prev}). It relies on the fact that finding all predecessors (or all successors) can be done with a single DFS traversal, and hence it can be done in $O(|V|+|E|)$ time, which is, on average, $O(1)$ time per vertex (for bounded degree graphs).

Consider \textsc{tree-next} (Algorithm~\ref{alg:tree-next}) and \textsc{tree-prev} (Algorithm~\ref{alg:tree-prev}). Calling \textsc{tree-next}($v$) for all vertices $v$ corresponds to a full DFS run on $T$, which traverses each edge exactly twice. The average number of queries per call is therefore $O(1)$. The same applies to Algorithm~\ref{alg:tree-prev} \textsc{tree-prev}.
This immediately implies that for a random vertex $v$, finding a successor and predecessor of $v$ can be done in expected $O(1)$ time.
\end{proof}

Using \cref{fact:navigate-ordered-trees}, we can prove the following lemma.

\begin{lemma}[\textsc{dfs-next}]
\label{lem:dfs-next}
There is an algorithm \textsc{dfs-next} that for a given vertex $v \in V$:
\begin{itemize}
\item If the numbering of $G$ is a valid DFS numbering then either vertex $v+1$ is returned or \textsc{end-of-component} is returned if $v$ is the largest vertex in its connected component.
\item The expected (over vertices in $V$) number of queries to $G$ of algorithm \textsc{dfs-next} is $O(1)$.
\end{itemize}
There also exists an algorithm \textsc{dfs-prev} with corresponding properties.
\junk{
There is an algorithm \textsc{dfs-next} that, given $v \in  V$ behaves as follows.
\begin{itemize}
        \item It makes $O(1)$ queries about $G$ on average over $V$.
        \item If the numbering of $G$ is a valid DFS numbering then%
        \junk{
        \begin{itemize}
            \item vertex $v+1$ is returned or
            \item \textsc{end-of-component} is returned if $v$ is the largest vertex in its connected component.
        \end{itemize}
        }
        either vertex $v+1$ is returned or \textsc{end-of-component} is returned if $v$ is the largest vertex in its connected component.
\end{itemize}
There also exists an algorithm \textsc{dfs-prev} with corresponding properties.
}
\end{lemma}

\begin{proof}
Let $T$ be the ordered tree on $V \cup  \{0\}$ rooted at $0$ with the edges $\{p(v),v\}$ for $v \in  V$ and children ordered ascending by number. If the numbering of $G$ is a valid DFS numbering then (by \cref{claim:dfs-parent-is-p}) $T$ is precisely the DFS tree that gave rise to the numbering.
%, except for the introduction of the virtual parent $0$ of all orphans.
We can navigate to successors and predecessors in $T$ (by \cref{fact:navigate-ordered-trees}) with an expected number of $O(1)$ queries to $G$ per step. The only problem is that the virtual vertex $0$ cannot be accessed and has unbounded degree. Whenever we would have to access it, we return \textsc{end-of-component} instead, which is appropriate because a valid DFS backtracks to $0$ precisely if a connected component has been fully explored. (If the numbering of $G$ is invalid, the produced answer may be meaningless, but it is still obtained within the claimed expected time budget.)
\end{proof}

%------------------------------------------------------------------------

\subsection{Testing for conflicts}
\label{subsec:Testing-for-conflicts}

\cref{lem:dfs-characterization-epsilon-far} implies that an \eps-far instance has $\Omega(\eps n)$ pairwise distinct vertices and edges involved in conflicts.
%This immediately implies (see \cref{corollary:DFS-ub-weak}) a property tester detecting \eps-far instances with $\Omega(\sqrt{n/\eps})$ queries.
In this section, we will extend that approach and study separately \emph{local conflicting pairs} and \emph{global conflicting pairs} in order to improve the tester from \cref{corollary:DFS-ub-weak}. This notion depend on a locality parameter that we will later choose as $\ell  = n^{1/3}$.

\def\L#1{\textbf{\upshape\sffamily\color{black} (L#1)}}
\def\G{\textbf{\upshape (G)}}
\begin{definition}
    Let $(v,\{u,w\}) \in  V \subseteq  E$ with $p(v) < u < v < w$ be a conflicting pair. We speak of a \textbf{local conflict} of the following (not mutually exclusive) types
    \begin{description}
        \item[\L{1}] if $u-p(v) \le  \ell $ and $p(v) \neq  0$,
        \item[\L{2}] if $v-u \le  \ell $,
        \item[\L{3}] if $w-v \le  \ell $.
        %\Stefan{We might not need this case.}
    \end{description}
    We speak of a \textbf{global conflict} in all other cases i.e.,
    \begin{description}
        \item[\G] if $p(v) = 0$ and $\max\{v-u,w-v\} > \ell $; or if $\max\{u-p(v),v-u,w-v\} > \ell $.
    \end{description}
\end{definition}

Informally, local conflicts occur when $u$ is close (in the sense of its DFS number) to $p(v)$ or $v$, or $w$ is close $v$, in which case one can traverse the input graph to efficiently detect the conflict. For global conflicts, some more global approach will be needed.

%------------------------------------------------------------------------

\subsubsection{Testing for local conflicts}
\label{subsubsec:testing-local-conflicts}
%\paragraph{Testing for local conflicts.}
%
We can detect local conflicts by sampling vertices at random and walking forwards and backwards in the (supposed) DFS order using \cref{lem:dfs-next} as follows.

\begin{lemma}
\label{lem:testing-local-conflicts}
There is an algorithm that with $O(\ell/\eps)$ queries in expectation accepts all valid DFS numberings and rejects with probability %at least
$2/3$ if there is a matching $M \subseteq  V \subseteq  E$ of size at least $\frac{\eps n}{30}$ consisting of conflicting pairs of type \L1.
The same applies to types \L2 and \L3.
\end{lemma}

\begin{proof}
    Consider the following procedure \textsc{walk-from-}$p(v)$ that is given $v \in  V$ as input. It first checks if $p(v) \neq  0$. If so, it attempts to locate vertices $p(v)+1,p(v)+2,\dots $ using \textsc{dfs-next} from \cref{lem:dfs-next} until one of the following happens.
    \begin{itemize}
        \item If \textsc{dfs-next} reaches vertices out of order, then report the error.
        \item If \textsc{dfs-next} reaches a vertex $u$ that has a neighbour $w > v$ then report the error.
        \item If $v$ or $p(v)+\ell $ is reached without finding an error, then conclude that $v$ is not involved in a conflict of type \L1.
    \end{itemize}
    Clearly \textsc{walk-from-}$p(v)$ finds an error if $v$ is involved in a conflict of type \L1 and by \cref{lem:dfs-next} it performs $O(\ell )$ queries in expectation (over all $v \in V$). Our algorithm to detect conflicts of type \L1 is now simply to repeat \textsc{walk-from-}$p(v)$ for many $v$ sampled uniformly at random. Since the $\frac{\eps n}{30}$ vertices matched in $M$ are pairwise distinct, a random vertex from $V$ is involved in a conflict of type \L1 with probability $\frac{\eps }{30}$. If we sample $\frac{60}{\eps }$ vertices at random then the probability of sampling at least one $v$ involved in a conflict of type \L1 is at least
    \begin{align*}
        1-(1-\eps /30)^{60/\eps } 
        \ge  1- \mathrm{e}^{-(\eps /30) \cdot (60/\eps )} 
        = 1-\mathrm{e}^{-2}
        \ge  2/3 \enspace.
    \end{align*}
    The expected total number of queries amounts to $O(\ell /\eps)$, which concludes the claim.

    For conflicts of type \L2 a similar procedure \textsc{walk-backwards-from-}$v$ works. For conflicts of type \L3 a procedure \textsc{walk-forwards-from-}$v$ would \emph{not} work because \textsc{dfs-next} might report \textsc{end-of-component} before $w$ is reached (and without us being aware of $w$'s existence).
    A procedure \textsc{walk-backwards-from-}$w$ does work, however. Note that we have to sample $\frac{60d}{\eps }$ edges uniformly at random which is still $O(1/\eps )$ because $d$ is constant.
\end{proof}

%------------------------------------------------------------------------

\subsubsection{Testing for global conflicts}
\label{subsubsec:testing-global-conflicts}
%\paragraph{Testing for global conflicts.}
%
\cref{lem:testing-local-conflicts} provides an efficient tool to detect local conflicts but for global conflicts we use a different approach. We rely on \cref{lem:dfs-characterization-epsilon-far} that promises that there is a matching $M \subseteq  V \subseteq  E$ of size at least $\eps n/10$ consisting of conflicting pairs. Therefore, if most of conflicts defining the matching $M$ are global, we can use the following lemma.

\begin{lemma}
\label{lem:detecting-globally-conflicting}
There is an algorithm that performs $O(\sqrt{n/\ell }/\eps )$ queries in expectation and accepts all valid DFS numberings and that rejects with probability at least $2/3$ if there is a matching $M \subseteq  V \subseteq  E$ of size at least $\eps n/10$ consisting of conflicting pairs of type (G).
\end{lemma}

\begin{proof}
Let us partition matching $M$ into strips $M_1 \oplus M_2 \oplus \dots \oplus M_{\lceil \frac{n}{\ell} \rceil}$ according to the number of $u$, so that $M_j = \{(v,\{u,w\}) \in M: \lceil u/\ell\rceil = j\}$. Let us define sets $V_j := \{v \in V: \exists_{e \in E} \ (v,e) \in M_j\}$ and $E_j := \{e \in E: \exists_{v \in V} (v,e) \in M_j\}$. Let $\bar{v}_j$ be the median of $V_j$. Let $V_j^-$ be those vertices from $V_j$ that are at most $\bar{v}_j$ and $E_j^+$ those edges from $E_j$ matched to a vertex that is at least $\bar{v}_j$. Let $m_j := |V_j^-|$ and note that $|E_j^+| = m_j$ and that $m_j \ge \frac{|M_j|}{2}$.

We claim that \emph{every pair} in $V_j^- \times E_j^+$ is a conflicting pair. Indeed, let $v_1 \in V_j^-$, $\{u_2,w_2\} \in E_j^+$ with $u_2 < w_2$, and let $p_1 = (v_1, \{u_1,w_1\})$ and $p_2 = (v_2, \{u_2,w_2\})$ be the corresponding pairs in $M_j$. Then we observe the following:
\begin{align*}
    \def\reason#1#2{\stackrel{\text{(#1)}}{#2}}
    p(v_1)
    \reason{1}{\le} \max\{0,u_1-\ell\}
    \reason{2}{<} u_2
    \reason{2}{<} u_1+\ell
    \reason{1}{<} v_1
    \reason{3}{\le} \bar{v}_j
    \reason{3}{\le} v_2
    \reason{4}{<} w_2
    \enspace.
\end{align*}
To see this, we notice the following:
\begin{enumerate}[(1)]
\item $p_1$ is of type (G). Hence $u_1-p(v_1) > \ell$, unless $p(v_1) = 0$; moreover $v_1-u_1 > \ell$;
\item since both $p_1$ and $p_2$ are in $M_j$, we have $|u_1 - u_2| < \ell$;
\item by the definition of $V_j^-$ and $E_j^+$ the values of $v_1$ and $v_2$ fall on the respective side of the median $\bar{v_j}$;
\item $p_2$ is a conflicting pair.
\end{enumerate}

In view of the above, since every pair in $V_j^- \times E_j^+$ is a conflicting pair, we observe that over all $j$ we do not just have $m = |M|$ conflicting pairs to work with, but a collection of bicliques with $\sum_{i = 1}^{\lceil n/\ell\rceil} (m_j)^2$ conflicting pairs in total. If we happen to get hold of a $v \in V_j^-$ and an $e \in E_j^+$ for the same $j$, then we have a conflicting pair $(v,e)$.

In order to find a conflicting pair using the approach above, let us sample
%\Artur{For the analysis, one should use sampling with repetitions (since then the sampling is independent for each sampled vertex and each sampled edge), but for the algorithm, one could just sample $s$ random vertices and $s$ random edges (without repetition), and notice that the chance of finding a conflicting sample is not worse if we sampled with repetitions.}
(i.u.r.) $s$~vertices from $V$ and then (i.u.r.) $s$ edges from $E$, where $s$ will be set up momentarily.

For any $1 \le i, r \le s$, let $X_{i,r}$ to be the indicator random variable that the $i$th sampled element from $V$ and the $r$th sampled element from $E$ are (respectively) from the sets $V_j^-$ and $E_j^+$ for the same $j$. Observe that if we define random variable $X$ as
\junk{
\begin{align*}
    X &= \sum_{i=1}^s\sum_{r=1}^s \Ex{X_{i,r}}
    \enspace,
\end{align*}
}
$X = \sum_{i=1}^s\sum_{r=1}^s \Ex{X_{i,r}}$,
then by the arguments above, if $X$ is positive then we have detected a conflicting pair. Using the second moment method, we can prove the following lemma.

\begin{lemma}
\label{lemma:bound-for-prob}
Let $s \ge c \sqrt{\frac{n^3}{\ell \cdot |M|^2}}$ for a sufficiency large constant $c$ and let $s = \omega((d/\eps)^3)$. Then $\Pr[X>0] \ge 0.99$.
\end{lemma}

\begin{proof}
The proof is a rather straightforward, though also rather lengthly application of the second moment method.

We will use the second moment method
\begin{align}
\label{ineq:final-analysis-2nd-moment}
    \Pr[X > 0] &\ge \frac{\Ex{X}^2}{\Ex{X^2}}
    \enspace.
\end{align}
We first give a lower bound (\ref{ineq:final-analysis-EX}) for $\Ex{X}$ and then an upper bound (\ref{ineq:final-analysis-very-key-bound}) for $\Ex{X^2}$ showing that $\Ex{X^2} \le \Ex{X} + (1+o(1)) \cdot \Ex{X}^2$. By (\ref{ineq:final-analysis-2nd-moment}), combining these two results will yield \cref{lemma:bound-for-prob}.

\subparagraph*{Bounding $\Ex{X}$.}
We begin with the study of $\Ex{X}$. Observe that
\begin{align*}
    \Ex{X_{i,r}} &=
    \sum_{j=1}^{\lceil n/\ell \rceil} \frac{|V_j^-| \cdot |E_j^+|}{n \cdot |E|} =
    \frac{1}{n \cdot |E|} \sum_{j=1}^{\lceil n/\ell \rceil} (m_j)^2
    \enspace,
\end{align*}
and therefore
\begin{align*}
    \Ex{X} &=
    \Ex{\sum_{i=1}^s\sum_{r=1}^s X_{i,r}} =
    \sum_{i=1}^s\sum_{r=1}^s \Ex{X_{i,r}} =
    \frac{s^2}{n \cdot |E|} \sum_{j=1}^{\lceil n/\ell \rceil} (m_j)^2
    \enspace.
\end{align*}
Next, by Cauchy-Schwarz inequality, % and using the fact that $m_j \ge |M_j|/2$,
if we set $m := \sum_{j=1}^{\lceil n/\ell \rceil} m_j$, we can bound this as follows:
\begin{align}
\label{ineq:final-analysis-EX}
    \Ex{X} &=
    \frac{s^2}{n \cdot |E|} \sum_{j=1}^{\lceil n/\ell \rceil} (m_j)^2
        \ge
    \frac{s^2}{n \cdot |E| \cdot \lceil n/\ell \rceil} \left(\sum_{j=1}^{\lceil n/\ell \rceil} m_j \right)^2
        \ge
    \frac{s^2 \cdot m^2}{n \cdot |E| \cdot \lceil n/\ell \rceil}
%        \ge
%    \frac{s^2 \cdot |M|^2}{4 \cdot n \cdot |E| \cdot \lceil n/\ell \rceil}
    \enspace.
\end{align}

\subparagraph*{Bounding $\Ex{X^2}$.}
Next, we will bound $\Ex{X^2}$. We have the following,
\begin{align}
\notag
    \Ex{X^2} &=
    \sum_{\substack{1 \le i, i' \le s\\1\le r, r' \le s}} \Ex{X_{i,r} X_{i',r'}}
        \\
\notag
        & =
    \sum_{\substack{1 \le i \le s\\1\le r \le s}} \Ex{X_{i,r} X_{i,r}} +
    \!\!\!\!\!\sum_{\substack{1 \le i \ne i' \le s\\1\le r \le s}} \!\!\!\!\!\Ex{X_{i,r} X_{i',r}} +
    \!\!\!\!\!\sum_{\substack{1 \le i \le s\\1\le r \ne r' \le s}} \!\!\!\!\!\Ex{X_{i,r} X_{i,r'}} +
    \!\!\!\!\!\sum_{\substack{1 \le i \ne i' \le s\\1\le r \ne r' \le s}} \!\!\!\!\!\Ex{X_{i,r} X_{i',r'}}
        \\
\notag
        &=
    \Ex{X} +
    \sum_{\substack{1 \le i \ne i' \le s\\1\le r \le s}} \Ex{X_{i,r} X_{i',r}} +
    \sum_{\substack{1 \le i \le s\\1\le r \ne r' \le s}} \Ex{X_{i,r} X_{i,r'}} +
    \sum_{\substack{1 \le i \ne i' \le s\\1\le r \ne r' \le s}} \Ex{X_{i,r}} \Ex{X_{i',r'}}
        \\
        &\le
    \Ex{X} +
    \sum_{\substack{1 \le i \ne i' \le s\\1\le r \le s}} \Ex{X_{i,r} X_{i',r}} +
    \sum_{\substack{1 \le i \le s\\1\le r \ne r' \le s}} \Ex{X_{i,r} X_{i,r'}} +
    \Ex{X}^2
\label{ineq:final-analysis-EX2}
    \enspace.
\end{align}

Let $v_i$ denote the $i$th sampled vertex and $e_r$ the $r$th sampled edges. To bound the second and the third terms in the summation, we observe that
\begin{align*}
    \Ex{X_{i,r} X_{i',r}} &=
    \Pr[\exists _{1 \le k \le \lceil n/\ell \rceil} v_i, v_{i'} \in V_k^-, e_r \in E_k^+] =
    \sum_{k=1}^{\lceil n/\ell \rceil} \frac{|V_k^-|^2 \cdot |E_k^+|}{n^2 \cdot |E|} =
    \sum_{k=1}^{\lceil n/\ell \rceil} \frac{(m_k)^3}{n^2 \cdot |E|}
    \enspace.
        \\
    \Ex{X_{i,r} X_{i,r'}} &=
    \Pr[\exists _{1 \le k \le \lceil n/\ell \rceil} v_i \in V_k^-, e_r, e_{r'} \in E_k^+] =
    \sum_{k=1}^{\lceil n/\ell \rceil} \frac{|V_k^-| \cdot |E_k^+|^2}{n \cdot |E|^2} =
    \sum_{k=1}^{\lceil n/\ell \rceil} \frac{(m_k)^3}{n \cdot |E|^2}
    \enspace.
\end{align*}
\junk{
This immediately gives the following,
\begin{align*}
    \Ex{X^2} &\le
%    \Ex{X} +
%    \sum_{\substack{1 \le i \ne i' \le s\\1\le r \le s}} \Ex{X_{i,r} X_{i',r}} +
%    \sum_{\substack{1 \le i \le s\\1\le r \ne r' \le s}} \Ex{X_{i,r} X_{i,r'}} +
%    \Ex{X}^2
%        \\
%        &=
    \Ex{X} +
    \frac{s^3}{n^2 \cdot |E|} \sum_{k=1}^{\lceil n/\ell \rceil} (m_k)^3 +
    \frac{s^3}{n \cdot |E|^2} \sum_{k=1}^{\lceil n/\ell \rceil} (m_k)^3 +
    \Ex{X}^2
    \enspace.
\end{align*}
}
Further, observe that since\footnote{To see that, observe that our definition of sets $M_j$ ensures that every edge in $M_j$ is of the form $\{u,w\}$ with $(j-1) \ell < u \le j \ell$ and each vertex $u$ is of degree at most $d$.} $0 \le m_k \le d\ell$ for each $m_k$, and since $\sum_{k=1}^{\lceil n/\ell \rceil} m_k = m$, we have
\begin{align*}
    \sum_{k=1}^{\lceil n/\ell \rceil} (m_k)^3 &\le
    \sum_{k=1}^{\lceil m/d\ell \rceil} (m_k)^3 \le
    \lceil m/d\ell \rceil \cdot (d\ell)^3 \le
    2 m d^2 \ell^2
    \enspace.
\end{align*}
Therefore, if we plug these bounds in our upper bound of $\Ex{X^2}$ in (\ref{ineq:final-analysis-EX2}) then we obtain
\begin{align}
%\notag
    \Ex{X^2} &\le
%    \Ex{X} +
%    \frac{s^3}{n^2 \cdot |E|} \sum_{k=1}^{\lceil n/\ell \rceil} (m_k)^3 +
%    \frac{s^3}{n \cdot |E|^2} \sum_{k=1}^{\lceil n/\ell \rceil} (m_k)^3 +
%    \Ex{X}^2
%        \\
%        &\le
%    \Ex{X} +
%    \frac{2 m d^2 \ell^2 s^3}{n^2 \cdot |E|} +
%    \frac{2 m d^2 \ell^2 s^3}{n \cdot |E|^2} +
%    \Ex{X}^2
%        \\
%        &\le
    \Ex{X} +
    \frac{2 m d^2 \ell^2 s^3}{n \cdot |E|} \cdot \left(\frac1n + \frac1{|E|}\right) +
    \Ex{X}^2
\label{ineq:final-analysis-key-bound}
    \enspace.
\end{align}
We claim that for $s = \omega((d/\eps)^3)$, the middle term is $o(\Ex{X}^2)$. Observe the following,

\begin{align*}
    \frac{2 m d^2 \ell^2 s^3}{n \cdot |E|} \cdot \left(\frac1n + \frac1{|E|}\right) &=
    \frac{2 m d^2 \ell^2 s^3 \cdot (n+|E|)}{n^2 \cdot |E|^2} \le
    \frac{2 m d^2 \ell^2 s^3 \cdot (dn)}{n^2 \cdot |E|^2} =
%    \frac{2 m d^3 \ell^2 s^3}{n \cdot |E|^2} =
    \frac{4 s^4 m^4 \ell^2}{n^4 |E|^2} \cdot \frac{d^3 n^3}{2 s m^3}
    \enspace.
\end{align*}
\junk{
Observe that $m \ge |M|/2 \ge \eps n / 20$, and hence, if for some $\alpha$ we have $s \ge \frac{4 \cdot 10^3 (d/\eps)^3}{\alpha}$, then
\begin{align*}
    \frac{d^3 n^3}{2 s m^3} &\le
    \frac{d^3 n^3}{2 \frac{4 \cdot 10^3 (d/\eps)^3}{\alpha} (\eps n/20)^3} =
    \alpha
    \enspace.
\end{align*}
}
Therefore, if $s = \omega((d/\eps)^3)$, then, by (\ref{ineq:final-analysis-EX}), we obtain
\begin{align*}
    \frac{2 m d^2 \ell^2 s^3}{n \cdot |E|} \cdot \left(\frac1n + \frac1{|E|}\right) &\le
    \frac{4 s^4 m^4 \ell^2}{n^4 |E|^2} \cdot \frac{d^3 n^3}{2 s m^3} =
    o(\Ex{X}^2)
    \enspace.
\end{align*}
This combined with (\ref{ineq:final-analysis-key-bound}) immediately implies that for $s = \omega((d/\eps)^3)$, it holds,
\begin{align}
%\notag
    \Ex{X^2} &\le
    \Ex{X} +
    \frac{2 m d^2 \ell^2 s^3}{n \cdot |E|} \cdot \left(\frac1n + \frac1{|E|}\right) +
    \Ex{X}^2
        \le
    \Ex{X} +
    (1+o(1)) \cdot \Ex{X}^2
\label{ineq:final-analysis-very-key-bound}
    \enspace.
\end{align}

To conclude the proof of \cref{lemma:bound-for-prob}, we take $s = \omega((d/\eps)^3)$ and put (\ref{ineq:final-analysis-very-key-bound}) in (\ref{ineq:final-analysis-2nd-moment}), to obtain
\begin{align*}
    \Pr[X > 0] &\ge
    \frac{\Ex{X}^2}{\Ex{X^2}} \ge
    \frac{\Ex{X}^2}{\Ex{X} + (1+o(1)) \cdot \Ex{X}^2} =
    \frac{1}{1 + o(1) + 1/\Ex{X}}
    \enspace.
\end{align*}
Therefore, if $\Ex{X} \ge 100$, then $\Pr[X > 0] \ge 0.99$.
\junk{
\begin{align*}
    \frac{1}{1 + o(1) + \frac{1}{100}} &\ge
    \frac{1}{1 + \frac{1}{99}} =
    \frac{99}{100}
    \enspace.
\end{align*}
}

Now, by (\ref{ineq:final-analysis-EX}), we have the following,
\begin{align*}
    \Ex{X} &\ge
    \frac{s^2 \cdot m^2}{n \cdot |E| \cdot \lceil n/\ell \rceil}
        \ge
    \frac{s^2 \cdot |M|^2}{4 \cdot n \cdot |E| \cdot \lceil n/\ell \rceil}
        \ge
    \frac{s^2 \cdot |M|^2}{4 \cdot n \cdot (dn/2) \cdot (n/\ell)}
        =
    \frac{s^2 \cdot \ell \cdot |M|^2}{2 \cdot d \cdot n^3}
    \enspace.
\end{align*}
\junk{
\begin{align*}
    \Ex{X} &\ge
    \frac{s^2 \cdot m^2}{n \cdot |E| \cdot \lceil n/\ell \rceil}
        \ge
    \frac{s^2 \cdot |M|^2}{4 \cdot n \cdot |E| \cdot \lceil n/\ell \rceil}
        \ge
    \frac{s^2 \cdot (\eps n / 10)^2}{4 \cdot n \cdot (dn/2) \cdot (n/\ell)}
        =
    \frac{s^2 \cdot \eps^2}{200 \cdot d \cdot n^{2/3}}
    \enspace.
\end{align*}
}
Therefore, if $s \ge \sqrt{\frac{200 \cdot d \cdot n^3}{\ell \cdot |M|^2}}$ and using the fact that $s = \omega((d/\eps)^3)$, then $\Ex{X} \ge 100$, completing the proof of \cref{lemma:bound-for-prob}.
\end{proof}

Now, we are ready to complete the proof of \cref{lem:detecting-globally-conflicting}. If we sample i.u.r. at least $c \cdot \sqrt{\frac{n^3}{\ell |M|^2}}$ vertices from $V$ and at least $c \cdot \sqrt{\frac{n^3}{\ell |M|^2}}$ edges from $E$ for a sufficiently large constant $c$, then by \cref{lemma:bound-for-prob}, with a constant probability we will detect a conflicting vertex. At the same time, if the DFS numbering is valid then the algorithm will accept it. This yields \cref{lem:detecting-globally-conflicting} since in our setting $|M| = \Omega(\eps n)$ and $\ell = n^{1/3}$.
\end{proof}

%------------------------------------------------------------------------

\subsection{Putting all together: The proof of Theorem \ref{thm:DFS-ub}}
\label{subsec:proof-of-thm:DFS-ub}

Now we are ready to complete the proof of \cref{thm:DFS-ub}.
Our algorithm runs the algorithms from \cref{lem:testing-local-conflicts} (responsible or conflicts of type \L1, \L2 and \L3) and from \cref{lem:detecting-globally-conflicting} one after the other. If any of them rejects $G$ then we reject $G$, otherwise we accept $G$.
The \emph{expected} total running time is $O(\ell /\eps +\sqrt{n/\ell }/\eps )$, which
%motivates setting $\ell  = \Theta(n^{1/3})$ to obtain $O(n^{1/3}/\eps )$ as claimed.
with $\ell  = \Theta(n^{1/3})$ gives $O(n^{1/3}/\eps )$ as claimed. (The query complexity can be made $O(n^{1/3}/\eps )$ using Markov inequality.)

Concerning correctness, it is clear that instances with valid DFS numberings are always accepted. If $G$ is $\eps $-far from a valid DFS numbering then by \cref{lem:dfs-characterization-epsilon-far} there is a matching $M \subseteq  V \subseteq  E$ of $|M| \ge  \eps n/5$ conflicting pairs. Since each matching edge falls into (at least) one type, at least one of the following statements holds.
\begin{itemize}
    \item There is a matching $M_{\L1}$ of $|M_{\L1}|\ge  \eps n/30$ conflicting pairs of type \L1.
    \item There is a matching $M_{\L2}$ of $|M_{\L2}|\ge  \eps n/30$ conflicting pairs of type \L2.
    \item There is a matching $M_{\L3}$ of $|M_{\L3}|\ge  \eps n/30$ conflicting pairs of type \L3.
    \item There is a matching $M_{\G}$ of $|M_{\G}|\ge  \eps n/10$ conflicting pairs of type \G.
\end{itemize}
In each case, the corresponding algorithm rejects $G$ with probability $2/3$ so overall we reject $G$ with probability at least $2/3$.
\qed

%-------------------------------------------------------------------------------------------------------

\subsection{Achieving running time of \texorpdfstring{$O(n^{1/3}/\eps)$}{Omega(n**(1/3)/\eps)} in Theorem~\ref{thm:DFS-ub}}
%{Deciding if a set of vertices and edges induces a conflict}
\label{sec:sweep-line-decider}

The following lemma implies immediately that our main algorithm from \cref{thm:DFS-ub} does not merely have a \emph{query complexity} of $O(n^{1/3}/\eps)$, but also a \emph{running time} of $O(n^{1/3}/\eps)$.

\begin{lemma}
There is an algorithm that, given $V' \subseteq  V$ and $E' \subseteq  E$ returns
\begin{itemize}
\item returns \textsc{conflict} if there exist $v_1,v_2 \in  V'$ such that $(v_1,\{p(v_2),v_2\})$ is a conflicting pair,
\item returns \textsc{conflict} if there is a conflicting pair in $V' \subseteq  E'$,
\item returns \textsc{no conflict} otherwise.
\end{itemize}
The algorithm runs in time $\mathrm{sort}(O(|V'|+|E'|))$ where $\mathrm{sort}(k)$ is the time needed to sort $k$ elements of $[n]$.
\end{lemma}

\begin{proof}
In a manner resembling sweepline algorithms, we intend to “sweep” through the range $[0,n]$ from left to right handling relevant events in sorted order.

For each $v \in  V'$ we consider a \emph{vertex-interval} $[p(v),v] \subseteq  \mathbb{Z}$ and for each $\{u<w\} \in  E'$ an \emph{edge-interval} $[u,w] \subseteq  \mathbb{Z}$. For each such interval $[x,y]$ there is a corresponding start-event at $x$ and an end-event $y$, which we sort in increasing order.\footnote{For the ensuing approach, ties need to be broken as follows: End events come before start events. Two end events or two start events are sorted in \emph{decreasing} order of their sibling events.}

While iterating over the events we maintain a sorted list $L$ of \emph{active} $v \in  V'$ for which the start event but not end event of $[p(v),v] $has been processed. Processing the end event of a vertex interval $[p(v),v]$ (at time $v$) just requires removing $\min L$ from $L$ (in $O(1)$ time). When processing the start event of $[p(v),v]$ (at time $p(v)$) we check if $v < \min L$. If so we can add $v$ as the new minimum of $L$ (in $O(1)$ time). Otherwise we have $p(\min(L)) < p(v) < \min(L) < v$ and can return \textsc{conflict}. Likewise, when processing the start event of an edge interval $[u,w]$ we check if $w > \min(L)$ and if so can return \textsc{conflict} due to $p(\min(L)) < u < \min(L) < w$. Otherwise we know that any $v \in  V'$ with $p(v) < u < v$ must be active and therefore satisfy $w \le  \min(L) \le  v$ so $\{u,w\} \in  E'$ is not in conflict with any $v \in  V'$. For end events of edge-intervals no action in required. If all events are processed and no \textsc{conflict} is reported, we can correctly report \textsc{no conflict}.

It is easy to see that the running time of our algorithm is dominated by sorting.
\end{proof}

%-------------------------------------------------------------------------------------------------------

%\includepdf[pages=-,scale=1.5,offset=-11mm 0mm]{DFS-upper-bound}

%-------------------------------------------------------------------------------------------------------

%\newpage
\section{Conclusions}
\label{section:conclusions}

In this paper we introduced a variant of the standard bounded-degree graph model in the property testing setting that works for labeled graphs and allows also label queries. We demonstrated the strength of the model on our new study of DFS numbering. Our main technical contribution is a tight analysis for detecting whether the input labeled graph is properly DFS numbered or it is \eps-far from having a valid DFS numbering. We demonstrated that this task can be solved with $\Omega(n^{1/3}/\eps)$ queries and also $\Omega(n^{1/3})$ queries are necessary.

We observe that while our analysis is presented for undirected graphs, similar arguments hold also for \textbf{\emph{directed graphs}}. The lower bound from \cref{thm:DFS-lb} trivially extends to directed graphs and a careful pass through the algorithm in \cref{thm:DFS-ub} shows that the analysis can be extended accordingly. However, to implement this algorithm efficiently in our setting, we need to allow access to incoming and outgoing edges. (See Appendix~\ref{section:directed-DFS} for more details.)

Our analysis can be also extended (but only for undirected graphs) to the \textbf{\emph{DFS finishing numbers}} (FIN-numberings \cite{CLRS22}). 
In Appendix~\ref{sec:FIN-numbering}, we show that (for undirected graphs) a numbering $\num : V \rightarrow [n]$ is a valid FIN numbering iff the reverse numbering $\mun$ with $\mun(i) = n+1-i$ is a valid DFS numbering. This immediately implies that our results for testing valid DFS numberings extend to testing valid FIN numberings in undirected~graphs.

%Let us also notice that our analysis can be also extended (but only for undirected graphs) to the DFS finishing numbers (FIN-numberings, see also Chapter~20.3 in \cite{CLRS22}). Indeed, in \cref{sec:FIN-numbering}, we sketch the arguments showing that (for undirected graphs) a numbering $\num : V \rightarrow [n]$ is a valid FIN numbering if and only if the reverse numbering $\mun$ with $\mun(i) = n+1-i$ is a valid DFS numbering. This immediately implies that our upper and lower bounds for testing valid DFS numberings (\cref{thm:DFS-lb,thm:DFS-ub}) carry over to testing valid FIN numberings in undirected graph.
%(We also notice that the claim (about the relationship between $\num$ and $\mun$) fails for directed graphs, see \cref{sec:FIN-numbering}.)% For example, this follows easily from Exercise 20-3-5 in \cite{CLRS22}.

%===========================================================================
%= REFERENCES
%===========================================================================

%\bibliographystyle{plainurl}
\bibliographystyle{alpha}
\bibliography{Czumaj-Sohler-Walzer-ESA-2025-references}

\newcommand{\etalchar}[1]{$^{#1}$}
\begin{thebibliography}{HKNO09}

\bibitem[AFNS09]{AFNS09}
Noga Alon, Eldar Fischer, Ilan Newman, and Asaf Shapira.
\newblock A combinatorial characterization of the testable graph properties:
  It's all about regularity.
\newblock {\em SIAM Journal on Computing}, 39(1):143--167, 2009.

\bibitem[AKP24]{AKP21}
Isolde Adler, Noleen K{\"{o}}hler, and Pan Peng.
\newblock On testability of first-order properties in bounded-degree graphs and
  connections to proximity-oblivious testing.
\newblock {\em SIAM Journal on Computing}, 53(4):825--883, 2024.

\bibitem[BR02]{BR02}
Michael~A. Bender and Dana Ron.
\newblock Testing properties of directed graphs: Acyclicity and connectivity.
\newblock {\em Random Structures \& Algorithms}, 20(2):184--205, 2002.

\bibitem[BSS10]{BSS10}
Itai Benjamini, Oded Schramm, and Asaf Shapira.
\newblock Every minor-closed property of sparse graphs is testable.
\newblock {\em Advances in Mathematics}, 223(6):2200--2218, 2010.

\bibitem[BY22]{BY22}
Arnab Bhattacharyya and Yuichi Yoshida.
\newblock {\em Property Testing --- Problems and Techniques}.
\newblock Springer Verlag, 2022.

\bibitem[CGR{\etalchar{+}}14]{CGRSSS14}
Artur Czumaj, Oded Goldreich, Dana Ron, C.~Seshadhri, Asaf Shapira, and
  Christian Sohler.
\newblock Finding cycles and trees in sublinear time.
\newblock {\em Random Structures \& Algorithms}, 45(2):139--184, 2014.

\bibitem[CLRS22]{CLRS22}
Thomas~H. Cormen, Charles~E. Leiserson, Ronald~L. Rivest, and Clifford Stein.
\newblock {\em Introduction to Algorithms, 4th Edition}.
\newblock {MIT} Press and McGraw-Hill, 2022.

\bibitem[CPS16]{CPS16}
Artur Czumaj, Pan Peng, and Christian Sohler.
\newblock Relating two property testing models for bounded degree directed
  graphs.
\newblock In {\em Proceedings of the 48th Annual ACM Symposium on Theory of
  Computing (STOC)}, pages 1033--1045, 2016.

\bibitem[CSS09]{CSS09}
Artur Czumaj, Asaf Shapira, and Christian Sohler.
\newblock Testing hereditary properties of nonexpanding bounded-degree graphs.
\newblock {\em SIAM Journal on Computing}, 38(6):2499--2510, 2009.

\bibitem[FPS19]{FPS19}
Hendrik Fichtenberger, Pan Peng, and Christian Sohler.
\newblock Every testable (infinite) property of bounded-degree graphs contains
  an infinite hyperfinite subproperty.
\newblock In {\em Proceedings of the 30th Annual {ACM-SIAM} Symposium on
  Discrete Algorithms (SODA)}, pages 714--726, 2019.

\bibitem[GGR98]{GGR98}
Oded Goldreich, Shafi Goldwasser, and Dana Ron.
\newblock Property testing and its connection to learning and approximation.
\newblock {\em Journal of the ACM}, 45(4):653--750, 1998.

\bibitem[Gol17]{Goldreich17}
Oded Goldreich.
\newblock {\em Introduction to Property Testing}.
\newblock Cambridge University Press, 2017.

\bibitem[GR98]{GR98}
Oded Goldreich and Dana Ron.
\newblock A sublinear bipartiteness tester for bunded degree graphs.
\newblock In {\em Proceedings of the 30th Annual ACM Symposium on Theory of
  Computing (STOC)}, pages 289--298, 1998.

\bibitem[GR02]{GR02}
Oded Goldreich and Dana Ron.
\newblock Property testing in bounded degree graphs.
\newblock {\em Algorithmica}, 32(2):302--343, 2002.

\bibitem[GR11]{GR11}
Oded Goldreich and Dana Ron.
\newblock On proximity-oblivious testing.
\newblock {\em SIAM Journal on Computing}, 40(2):534--566, 2011.

\bibitem[HKNO09]{HKNO09}
Avinatan Hassidim, Jonathan~A. Kelner, Huy~N. Nguyen, and Krzysztof Onak.
\newblock Local graph partitions for approximation and testing.
\newblock In {\em Proceedings of the 50th Annual {IEEE} Symposium on
  Foundations of Computer Science (FOCS)}, pages 22--31, 2009.

\bibitem[HT74]{HT74}
John~E. Hopcroft and Robert~E. Tarjan.
\newblock Efficient planarity testing.
\newblock {\em Journal of the ACM}, 21(4):549--568, 1974.

\bibitem[KT13]{KT13}
Jon~M. Kleinberg and {\'{E}}va Tardos.
\newblock {\em Algorithm Design}.
\newblock Addison-Wesley, 2013.

\bibitem[Luc82]{Lucas1882}
{\'{E}}douard Lucas.
\newblock {\em Récreations Mathématiques}.
\newblock Librairie Albert Banchard, Paris, 1882.

\bibitem[NS13]{NS13}
Ilan Newman and Christian Sohler.
\newblock Every property of hyperfinite graphs is testable.
\newblock {\em SIAM Journal on Computing}, 42(3):1095--1112, 2013.

\bibitem[RRSS09]{RRSS09}
Sofya Raskhodnikova, Dana Ron, Amir Shpilka, and Adam~D. Smith.
\newblock Strong lower bounds for approximating distribution support size and
  the distinct elements problem.
\newblock {\em SIAM Journal on Computing}, 39(3):813--842, 2009.

\bibitem[RS96]{RS96}
Ronitt Rubinfeld and Madhu Sudan.
\newblock Robust characterizations of polynomials with applications to program
  testing.
\newblock {\em SIAM Journal on Computing}, 25(2):252--271, 1996.

\bibitem[Tar95]{Tarry1895}
G.~Tarry.
\newblock Le problème des labyrinthes.
\newblock {\em Nouvelles Annales de Mathématiques, 3e série}, 14:187--190,
  1895.

\bibitem[Tar72]{Tarjan72}
Robert~Endre Tarjan.
\newblock Depth-first search and linear graph algorithms.
\newblock {\em SIAM Journal on Computing}, 1(2):146--160, 1972.

\bibitem[Tar76]{T76}
Robert~Endre Tarjan.
\newblock Edge-disjoint spanning trees and depth-first search.
\newblock {\em Acta Informatica}, 6:171--185, 1976.

\bibitem[YI10]{YI10}
Yuichi Yoshida and Hiro Ito.
\newblock Testing {$k$}-edge-connectivity of digraphs.
\newblock {\em Journal of Systems Science and Complexity}, 23(1):91--101, 2010.

\end{thebibliography}

\addcontentsline{toc}{section}{References}

%===========================================================================

\appendix
%\newpage
\section*{\huge Appendix}
\addcontentsline{toc}{section}{Appendix}
%\newpage

%-------------------------------------------------------------------------------------------------------

\junk
{

\section{Further thoughts about the model}
\label{subsec:discussion-about-model}

In this section, we briefly discuss several possible natural modifications to our model used for testing DFS numbering.

\subsection{Further thoughts about the model: Modifying the labels}
\label{subsubsec:discussion-about-model-editing-labels}
Observe that in general, it might be natural to consider a revised definition of a labeling $\num$ being \eps-far from a valid DFS numbering, where while defining graph $G'$ in
%Definitions \ref{def:eps-far-DFS-numbering}, \ref{def:eps-far-DFS-numbering-max-deg}, and \ref{def:eps-far-DFS-numbering-max-deg-connected},
\cref{def:eps-far-DFS-numbering-max-deg}, in addition to edge modifications, one would allow also for modifications of the labels\footnote{We use the following revised definition: a labeling \emph{$\num$ is \eps-far from a valid DFS numbering} of $G$ if for any graph $G' = (V,E')$ of maximum degree at most $d$ with a valid DFS numbering $\num'$ we have $|E \triangle E'| + |\{ v \in V: \num(v) \ne \num'(v)\}| \ge \eps n$.}. We observe that for bounded degree graphs, adding the labels modifications does not change the problem significantly.

\begin{lemma}
\label{lemma:two-models-editing-labels}
Let $\num$ be a permutation of $V$ to $\{1, \dots, n\}$. For a given bounded degree graph $G = (V,E)$, a labeling $\num$ is \eps-far from a valid DFS numbering according to the edges modifications if and only if $\num$ is $\Theta(\eps)$-far from a valid DFS numbering according to the edges and labels modifications.
\end{lemma}

\begin{proof}
Observe that the number of modifications of the edges to obtain a valid DFS numbering is not smaller than the number of modifications of the edges and the labels to obtain a valid DFS numbering. Therefore, if for a given bounded degree graph $G$, a labeling $\num$ is \eps-far from a valid DFS numbering according to the edges and labels modifications then $\num$ is also \eps-far from a valid DFS numbering according to the edges modifications.

For the other direction, if $\num$ is \eps-far from a valid DFS numbering according to the edges modifications then we can simulate modification of the labels by modification of the edges: to assign a given label to a vertex we just remove all its incident edges and add all edges incident to the vertex with the label sought. (Observe that a vertex might have had some edges of its own that have to be removed but by correcting the labels one by one, we can ensure that once we fix vertex $v$ by assigning it to label $i$ with vertex $u$ having label $i$ before, then later we will have to fix the label of vertex $u$, resolving the issue.)
%
%\Stefan{But that vertex might have had some edges of its own that have to be removed... I think we cannot really do label assignments we can only do label swaps. Swapping two labels can be done in $\le  2d$ steps, hence swapping labels is no more powerful than edge edits. Moreover, a sequence of $k$ label assignments that starts and ends with no repeated labels can be simulated with $\le  k$ label swaps.}
%\Artur{Not really. We correct the labels one by one, and so once we fix vertex $v$ by assigning it to label $i$ with vertex $u$ having label $i$ before, then later we will have to fix the label of vertex $u$. So at the end, all will work well.}
Observe that if we apply this operation to all vertices with the labels to be changed, then we do not increase the maximum degree, and hence, modifying of $k$ labels can be simulated by modifying at most $2dk$ edges. Therefore, if a labeling $\num$ is \eps-far from a valid DFS numbering according to the edges modifications then $\num$ is also $(2d\eps)$-far from a valid DFS numbering according to the edges and labels modifications.
\end{proof}

%-------------------------------------------------------------------------------------------------------

\subsection{Why should the labels be a permutation?}
\label{subsubsec:discussion-about-model-non-permutations}
Our model assumes that the input labeling $\num$ is a permutation of $V$ to $\{1\,\dots,n\}$, but it may be natural to consider the case when $\num$ is not necessarily a permutation but rather an arbitrary function from $V$ to $\{1\,\dots,n\}$, allowing repetitions of labels. Observe that already the simple problem of testing if $\num$ is a permutation or is \eps-far from being a permutation (also known as the element distinctness problem) is known to require $\Theta(\sqrt{n}/\eps)$ queries (see, e.g., \cite{RRSS09}).
%\Artur{Christian: Is \cite{RRSS09} the right reference to the complexity of testing element distinctness?}
In view of that bound, the best what we could hope for without assuming that $\num$ is a permutation of $V$ to $\{1\,\dots,n\}$ is to achieve query complexity of $\Theta(\sqrt{n}/\eps)$. And this can be easily obtained by combining our algorithm in \cref{thm:DFS-ub} with the known algorithms testing element distinctness, leading to a testing algorithm with $\Theta(\sqrt{n}/\eps)$ queries.

%-------------------------------------------------------------------------------------------------------

}

%\junk
{

\section{Directed graphs: Testing DFS numbering}
\label{section:directed-DFS}

As we mentioned in \cref{section:conclusions}, our analysis in \cref{section:DFS-lb,section:DFS-ub} naturally extends to directed graphs.
Firstly, it is easy to see that the lower bound from \cref{thm:DFS-lb} trivially extends to directed graphs.
As for the upper bound, the key feature allowing us to consider directed graphs is that \cref{def:conflicting -pair} of conflicting pairs extends to directed graphs. Then, a careful pass through the algorithm in \cref{thm:DFS-ub} shows that the analysis can be extended accordingly. However, in order to implement it in our setting, for graph traversal, the tester requires access to all neighbors of any given vertex, both the endpoints of outgoing edges incident to a vertex and the endpoints of incoming edges incident to that vertex (in particular, in our analysis of local conflicts in \cref{subsubsec:testing-local-conflicts}).

For any vertex $v \in V$, let $\textsf{out}(v)$ denote the set of the endpoints of outgoing edges incident to $v$, $\textsf{out}(v) := \{u \in V: (v,u) \in E\}$, and $\textsf{in}(v)$ denote the set of the endpoints of incoming edges incident to $v$, $\textsf{in}(v) := \{u \in V: (u,v) \in E\}$. Then DFS numbering for directed graphs is defined by the following algorithm (the labeling is not unique and depends on the ordering in which the vertices are traversed in the for-loops of Algorithms~\ref{alg:DFS-numbering-def} and \ref{alg:DFS-numbering-def-rec}).

%------------------------------------------------------------------------

\def\dfsVisit{\textsc{dfs-visit}}
\begin{figure}[h]
\centerline{
\begin{minipage}[t]{0.47\textwidth}
\begin{algorithm}[H]
\caption{DFS numbering of $G$}
\label{alg:directed-DFS-numbering-def}%\small
\KwIn{
    directed graph $G = (V,E)$
}
\KwOut{
    %permutation
    $\num: V \rightarrow \{1,\dots,|V|\}$
}
    \BlankLine
    mark all vertices as undiscovered\;
    $t \leftarrow 1$ \tcp{\small smallest unused ID}
    \For(\tcp*[h]{\small arbitrary order}){$v \in  V$}{
        \If{$v$ is undiscovered}{
            \dfsVisit($v$)\;
        }
}
\end{algorithm}
\end{minipage}\qquad
\begin{minipage}[t]{0.53\textwidth}
\begin{algorithm}[H]%[htbp]
\caption{\dfsVisit($v$)}
\label{alg:directed-DFS-numbering-def-rec}%\small
    mark $v$ as discovered\;
    $\num(v) \leftarrow t$\;
    $t \leftarrow t +1$\;
    \For(\tcp*[h]{\small arbitrary order}){$u \in  \textsf{out}(v)$}{
        \If{$v$ is undiscovered}{
            \dfsVisit($v$)\;
        }
    }
\end{algorithm}
\end{minipage}
}
\end{figure}

%------------------------------------------------------------------------

We consider a setting where the access to the input graph $G=(V,E)$ with vertex labeling $\num$ is through the oracle allowing the following three types of queries:
\begin{itemize}
\item \textbf{In-neighbor queries:} for every vertex $v \in V$ and every $1 \le i \le d$, one can query the $i$th in-neighbor of vertex $v$ (i.e., the $i$th vertex in $\textsf{in}(v)$ with arbitrary but fixed order of vertices in set $\textsf{in}(v)$).
\item \textbf{Out-neighbor queries:} for every vertex $v \in V$ and every $1 \le i \le d$, one can query the $i$th out-neighbor of vertex $v$.
\item \textbf{Label queries:} for every vertex $v \in V$, one can query the label of $v$, $\num(v)$.
\end{itemize}

\begin{remark}
We notice that the requirement to access edges in both directions has been studied in the graph property testing setting (see, e.g., \cite{BR02,CPS16,Goldreich17,YI10}), though also the one-directional model has been considered (and testing in that model is known to be significantly more expensive, see, e.g., \cite{BR02,CPS16}).
\end{remark}

Then the following theorem extends \cref{thm:DFS-ub} to directed graphs.

\begin{theorem}
\label{thm:DFS-dub}
Let $0 < \eps < 1$. There is an algorithm that with oracle access to a labeled bounded degree \textbf{directed} graph $G$ on $n$ vertices, performs $O(n^{1/3}/\eps)$ queries to the oracle and accepts, if $G$ has a valid DFS numbering, and rejects with probability at least $\frac23$, if $G$ is \eps-far from having a valid DFS numbering.
\end{theorem}

The proof of \cref{thm:DFS-dub} mimics the proof of \cref{thm:DFS-ub}, and so we will sketch it here. In order to extend our analysis from \cref{section:DFS-ub} to directed graphs, we have to revise the definition of $p(v)$. For a directed implicitly labeled graph $G = (V, E)$ with $V = [n]$, we define $p(v)$ for $v \in V$ as
\begin{align*}
    p(v) :=
    \begin{cases}
        \max\{u \in \textsf{in}(v) \cap  [v-1]\} & \text{ if } \textsf{in}(v) \cap [v-1] \ne \emptyset \enspace,\\
        0 & \text{ otherwise.}
    \end{cases}
\end{align*}
Then, with the same analysis as for undirected graphs, \cref{claim:dfs-parent-is-p} holds for directed graphs and so does \cref{lem:dfs-characterization} with revised \cref{def:conflicting -pair} of conflicting pairs $(v,(u,w)) \in V \times E$ if $p(v) < u < v < w$.%
\junk{
Observe that the proof of \cref{lem:dfs-characterization} can be easily extended to directed graphs:

\begin{proof}[Proof of \cref{lem:dfs-characterization} for directed graphs]
\mbox{}
\begin{description}
\item[(i) $\rightarrow$ (ii).] Consider a DFS run that agrees with the DFS numbering and let $(v,(u,w)) \in V \times E$ be a pair with $p(v) < u < v$. By \cref{claim:dfs-parent-is-p}\Artur{Rewritten for directed graphs.}, $v$ was discovered from $p(v)$. Consider the time immediately before this happens. Vertex $u$ was already discovered (due to $u < v$) and is not active (since $p(v) < u$ is the end of the white path and $p(v)$ is the largest active vertex), meaning $u$ is finished. Therefore, since the finishing time of vertex $u$ is between $u$ and $w$, by a known property of DFS (see, e.g., \cite[Excercise~20.3-5]{CLRS22}), we cannot have edge $(u,w) \in E$ with $p(v) < u < v < w$. In particular $(v,(u,w))$ is not a conflicting pair.
\item [(ii) $\rightarrow$ (i).] We use induction on $k$, showing that there is a DFS run assigning number $v$ to each $v$, $1 \le v \le k$. For $k = 0$ the claim is trivial. Assume the claim holds for some $k-1$. Consider a corresponding DFS run that has just discovered vertex $k-1$; at that moment vertex $k$ is undiscovered. We have $0 \le  p(k) \le  k-1$. If $p(k)$ does not appear on the white path then we infer that because both $0$ and $k-1$ appear on the white path, there is a vertex $x$ on the white path with $p(x) < p(k) < x < k$. This makes $(x,(p(k),k))$ a conflicting pair, contradicting the assumption (ii). Therefore $p(k)$ does appear on the white path. There might be a vertex $u > p(k)$ further along on the white path. Let $w$ be a neighbor of $u$. Due to $p(k) < u < k$ we must have $w \le  k$, for otherwise $(k,(u,w))$ would be a conflicting pair. Moreover we have $w \neq  k$ by definition of $p(k)$. Thus $w \le  k-1$, meaning $w$ was already discovered. In particular the DFS will backtrack from active vertices until it reaches $p(k)$, from which it may legally choose to discover $k$ next. This concludes the induction.\qedhere
\end{description}
\end{proof}
}
Next, \cref{claim:how-to-fix} and \cref{lem:dfs-characterization-epsilon-far} work for directed graphs as well. Then, using the fact that one can access both in- and out-neighbors, \cref{fact:navigate-ordered-trees} and \cref{lem:dfs-next} hold without any change in the analysis, and so are \cref{lem:testing-local-conflicts,lem:detecting-globally-conflicting}. With all these claims at hand, as in \cref{subsec:proof-of-thm:DFS-ub}, we obtain the proof of \cref{thm:DFS-dub}.
\qed

Finally, let us notice that the arguments used in \cref{corollary:DFS-ub-weak} can be applied also for directed graphs. Since that approach does not need to use the neighbors queries in both directions, this is giving us a tester with $O(\sqrt{n/\eps})$ queries to the oracle that uses only the out-neighbor queries and label queries (and so, not using the in-neighbor queries).

%-------------------------------------------------------------------------------------------------------

\section{Testing FIN-numberings}
\label{sec:FIN-numbering}

%\Artur{Some of our claims should easily follow from \cite[Chapter~20.3]{CLRS22}. For example, Exercise 20-3-5 in \cite{CLRS22} implies that for directed graphs, for any cross-edge $(u,v) \in E$ we have $\num_{DFS}(u) < \num_{DFS}(v)$ and $\num_{FIN}(u) < \num_{FIN}(v)$, which is (I think) what you're essentially saying in the last paragraph of \cref{sec:FIN-numbering}. In particular: $\num_{DFS}(v) \ne n + 1 - \num_{FIN}(v)$.}%
%\Stefan{Feel free to replace any of my content with citations from appropriate places.}
%
In the context of a depth first search, the FIN numbering of vertices reflects the order in which vertices are marked as finished, i.e., $v$ receives FIN number $i$ if \dfsVisit($v$) is the $i$th call to \dfsVisit\ to conclude\footnote{To generate such numbers use Algorithm~\ref{alg:DFS-numbering-def} but within Algorithm~\ref{alg:DFS-numbering-def-rec} move lines 2 and 3 ($\num(v) \leftarrow t$ and $t \leftarrow t+1$) to the end, after line 6.}, see also Chapter~20.3 in \cite{CLRS22}. In this section, we will show that in undirected graphs a numbering $\num : V \rightarrow [n]$ is a valid FIN numbering if and only if the reverse numbering $\mun$ with $\mun(i) = n+1-i$ is a valid DFS numbering\footnote{
    This claim should not be confused with the mistaken claim that in a single DFS a vertex $v$ receiving DFS number $\num(v)$ receives FIN number $\mun(v)$. In general, the DFS run that produces the DFS numbering $\num$ visits vertices in a different order than the DFS run that produces the FIN numbering $\mun$.
}, which immediately implies that our upper and lower bounds for testing valid DFS numberings (\cref{thm:DFS-lb,thm:DFS-ub}) carry over to testing valid FIN numberings in undirected graph.

To prove the claim we consider ``mirrored'' notions of $p(v)$ and conflicting pairs from \cref{section:DFS-ub}. For any implicitly numbered graph we define
%\Artur{I'm not sure why the fact that in undirected graphs we have $\num_{DFS}(v) = n+1-\num_{FIN}(v)$ doesn't give us a tester at once? Why do we need to do more than just to refer to the DFS numbering tester (\cref{section:DFS-lb} for lower bound and \cref{section:DFS-ub} for upper bound)?}
%\Stefan{I have made some changes, saying that we \emph{do} obtain a tester (once the following Lemma is in place). I don't see a way around proving (i) $⇔$ (ii) though, if that's what you mean.}
\def\pfin{p^{\mathrm{FIN}}}
\[
    \pfin(v) :=
    \begin{cases}
        \min\{u \in N(v) \cap  \{v+1, \dots ,n\}\} & \text{ if $N(v) \cap  \{v+1, \dots ,n\} \neq  \emptyset$,}\\
        \infty & \text{ otherwise.}
    \end{cases}
\]
Note that the imagined virtual parent of all orphaned vertices is now a virtual vertex $\infty$ which intuitively finishes last.
We say a pair $(v,\{u,w\})$ for $v \in  V$ and $\{u,w\} \in  E$ is a \emph{fin-conflicting pair} if
\[
    w < v < u < \pfin(v).
\]
The following lemma summarizes the relationship between valid FIN-numberings and valid DFS numberings.
%It should be easy to see that a numbering has a FIN-conflicting pair if and only if the reverse numbering has a conflict pair.  Moreover:

\begin{lemma}
\label{lem:fin-numbering-characterisation}
Let $G = (V,E)$ be an undirected implicitly labeled graph. The following claims are equivalent.
\begin{enumerate}[\rm (i)]
\item $G$ has a valid FIN numbering.
\item There exists no FIN-conflicting pair.
\item There exists no conflicting pair in the reverse numbering.
\item The reverse numbering is a valid DFS numbering.
\end{enumerate}
\end{lemma}
The equivalence of (ii) and (iii) is easy to see since the definitions of \emph{FIN-conflicting pair} and \emph{conflicting pair} coincide except for flipping the ordering (as well as the sentinels $0$ and $\infty$). The equivalence of (iii) and (iv) was the content of \cref{lem:dfs-characterization}. Before we proceed to show the equivalence of (i) and (ii) we prove a claim analogous to \cref{claim:dfs-parent-is-p}.

\begin{claim}
\label{claim:fin-parent-is-pfin}
If the numbers of an implicitly labeled graph correspond to a FIN numbering then $\pfin(v)$ is the (possibly virtual) parent of $v$.
\end{claim}

\begin{proof}
If $v$ is an orphan then it has the highest FIN number in its connected component, hence $\{u \in N(v): u > v\} = \emptyset$ and thus $\pfin(v) = \infty$. This is also the virtual parent of $v$ by definition.

Otherwise $v$ has a non-virtual parent $x$. Since we only return to processing $x$ when $v$ is finished we have $x \in \{u \in N(v): u > v\}$. Now consider any $x' \in  \{u \in N(v): u > v\}$. Right after $v$ finishes, $x'$ must already be discovered but due to $x' > v$ it has not yet finished. Thus $x'$ is active and $\{u \in N(v): u > v\}$ only contains vertices on the white path. Since FIN numbers appear in descending order on the white path $x$ is the smallest among them as we defined it.
\end{proof}

\begin{proof}[Proof of \cref{lem:fin-numbering-characterisation}.]
\begin{description}
\item[(i) $\Rightarrow $ (ii).]
Consider a DFS run that agrees with the FIN numbering and let $(v,\{u,w\})$ for $v \in  V$ and $\{u,w\} \in  E$ be a pair with $w < v < \pfin(v)$ and $u > v$. By \cref{claim:fin-parent-is-pfin} $v$ was discovered from $\pfin(v)$. Consider the time immediately after $v$ has finished. Since $w < v$ has previously finished its neighbor $u \in  N(w)$ has been discovered, but since $u > v$ it has not yet finished. The active vertices $u$ and $\pfin(v)$ appear in descending order of FIN numbers along the white path with $\pfin(v)$ being the last one we have just returned to, so $\pfin(v) < u$. In particular $(v,\{u,w\})$ is not a FIN-conflicting pair.
\item[(ii) $\Rightarrow $ (i).]
\def\Tfin{T_{\mathrm{FIN}}}
Consider the tree $\Tfin = ([n] \cup  \{\infty\}, \{(\pfin(v),v) : v \in  [n]\}$ rooted at $\infty$. It is acyclic because numbers are descending along edges. Note that $\infty$ is the only vertex without incoming edges. We now establish two properties of $\Tfin$.

\emph{Firstly}, \emph{$G$ has no cross edges with respect to $\Tfin$}, meaning for any $\{u,w\} \in  E$ either $u$ is a decendant of $w$ or vice versa. To see this, assume without loss of generality that $w < u$. If $\pfin(w) = u$ there is nothing to show, and $\pfin(w) > u$ contradicts the definition of $\pfin(w)$ so we may assume $\pfin(w) < u$. Now consider the unique path $P$ from $\pfin(w)$ to $\infty$ in $\Tfin$. Since $\pfin(w) < u < \infty$ and numbers appear in increasing order on $P$ there must be two adjacent $v$ and $\pfin(v)$ on $P$ with $v < u \le  \pfin(v)$. Together we have $w < v < u \le  \pfin(v)$. The only way for this not to be a FIN-conflicting pair is $u = \pfin(v)$. Hence $u$ is on $P$ and thus an ancestor of $w$ as claimed.

\emph{Secondly}, \emph{subtrees of $\Tfin$ use non-overlapping ranges of number}, meaning the following. Consider a vertex $r$ in $\Tfin$ with two children $\ell $ and $r$ and let $T_\ell $ and $T_r$ be the trees of all descendants of $\ell $ and $r$, respectively. If $x_1$ and $x_2$ appear in $T_\ell $ and $y$ appears in $T_r$ then $x_1 < y < x_2$ is impossible. To see this consider for contradiction a counterexample with maximum sum $x_1 + y + x_2$. By maximality we have $x_2 = \ell $ (otherwise replace $x_2$ with $\ell $), $\pfin(y) > x_2$ (otherwise replace $y$ with $\pfin(y)$) and $\pfin(x_1) > y$ (otherwise replace $x_1$ with $\pfin(x_1)$). This gives $x_1 < y < \pfin(x_1) \le  \ell  = x_2 < \pfin(y)$ and thus the FIN-conflicting pair $(y,\{\pfin(x_1),x_1\})$.

We can now argue by induction on $k \in  [n]$, that there exists a DFS run and a point in time such that the run has assigned the first $k-1$ numbers as intended and where the white path is the unique path from $\infty$ to $k$ in $\Tfin$. For $k = 1$ note that there is a unique path from $\infty$ to $1$ in $\Tfin$ and a DFS run may legally follow it. Now assume we have established the claim for $k$ as witnessed by a DFS run currently stuck in $k$ and we want to extend the claim to $k+1$. First we show that the DFS may now legally backtrack from $k$. To see this consider $x \in  N(k)$. If $x < k$ then $x$ is finished by induction, so not undiscovered. if $x > k$ then by our first claim above $x$ must be an ancestor of $k$ in $\Tfin$. By induction $x$ must lie on the white path. Thus $x$ is already active and not undiscovered.
After our DFS run has backtracked to $\pfin(k)$ then, unless $\pfin(k) = k+1$, we still have to show that it can descend within $\Tfin$ to $k+1$. Since $k < k+1 < \pfin(k)$ our second claim shows that $k+1$ and $\pfin(k)$ do not reside in disjoint subtrees, in other words $\pfin(k)$ must be an ancestor of $k+1$. All vertices on the $\Tfin$-path from $\pfin(k)$ to $k+1$ are neither finished (because their numbers are $\ge  k+1$) nor active (because all active vertices are on the white path and ancestors of $\pfin(k)$), so those vertices are undiscovered and may legally be chosen by the DFS.\qedhere
\end{description}
\end{proof}

Unfortunately for directed graphs this close correspondence between valid DFS numberings and valid FIN numberings no longer holds. For instance for the graph $G = (\{x,y\},\{(x,y)\})$ with a single edge both numberings $\{ x \rightarrow 1, y \rightarrow 2\}$ and $\{ x \rightarrow 2, y\rightarrow 1\}$ are valid DFS numberings, but only the second is a valid FIN numbering. (Another way of arguing about it follows easily from Exercise 20-3-5 in \cite{CLRS22}.)

%-------------------------------------------------------------------------------------------------------

\section{Additional arguments supporting (\ref{DFS-lb-prop3-2}) in the proof of Lemma \ref{DFS-lb-prop3}}
\label{sec:proof:final-DFS-lb-prop3}

In Figure~\ref{fig-lb-3}, we present structural observations about an arm that can be revealed in $I_n(v_i)$. Here a \textcolor[rgb]{0.00,0.50,0.00}{green} solid dot indicates the joint that is contained in $I_n(v_i)$ and \textcolor[rgb]{0.00,0.07,1.00}{blue} circles indicate joints not contained in $I_n(v_i)$. The central fact is that each observation is consistent with exactly one good arm ($(G1)$ or $(G2)$) and exactly one bad arm ($(B1)$ or $(B2)$). In general, when we also take into account the case when $I_n(v_i)$ not contain any joint, we can say that $I_n(v_i)$ is consistent with the same number of good and bad arms. Therefore for any labeled graph $H$, we have $\Pr[I_n(v_i) = H | b=0] = \Pr[I_n(v_i) = H | b=1]$. This is precisely what it means that $I_n(v_i)$ and $b$ are independent. Since $G_n$ and $B_n$ select their arms independently, conditioning on \EPS, the same observation can be extended to hold for other vertices $v_i' \in V_R$, and hence to the entire graph $I_n(V_R^*)$. This implies that random variables $b$ and $I_n(V_R^*)$ are stochastically independent, which %in turn,
yields inequality (\ref{DFS-lb-prop3-2}) from the proof of \cref{DFS-lb-prop3} in \cref{subsection:key-fact-about-distinguishing}.

\begin{figure}[h]
%\centerline{\includegraphics[width=.3\textwidth]{DFS-LB-1}}
%\centerline{\includegraphics[width=.5\textwidth]{DFS-LB-5}}
\centerline{\includegraphics[width=.75\textwidth]{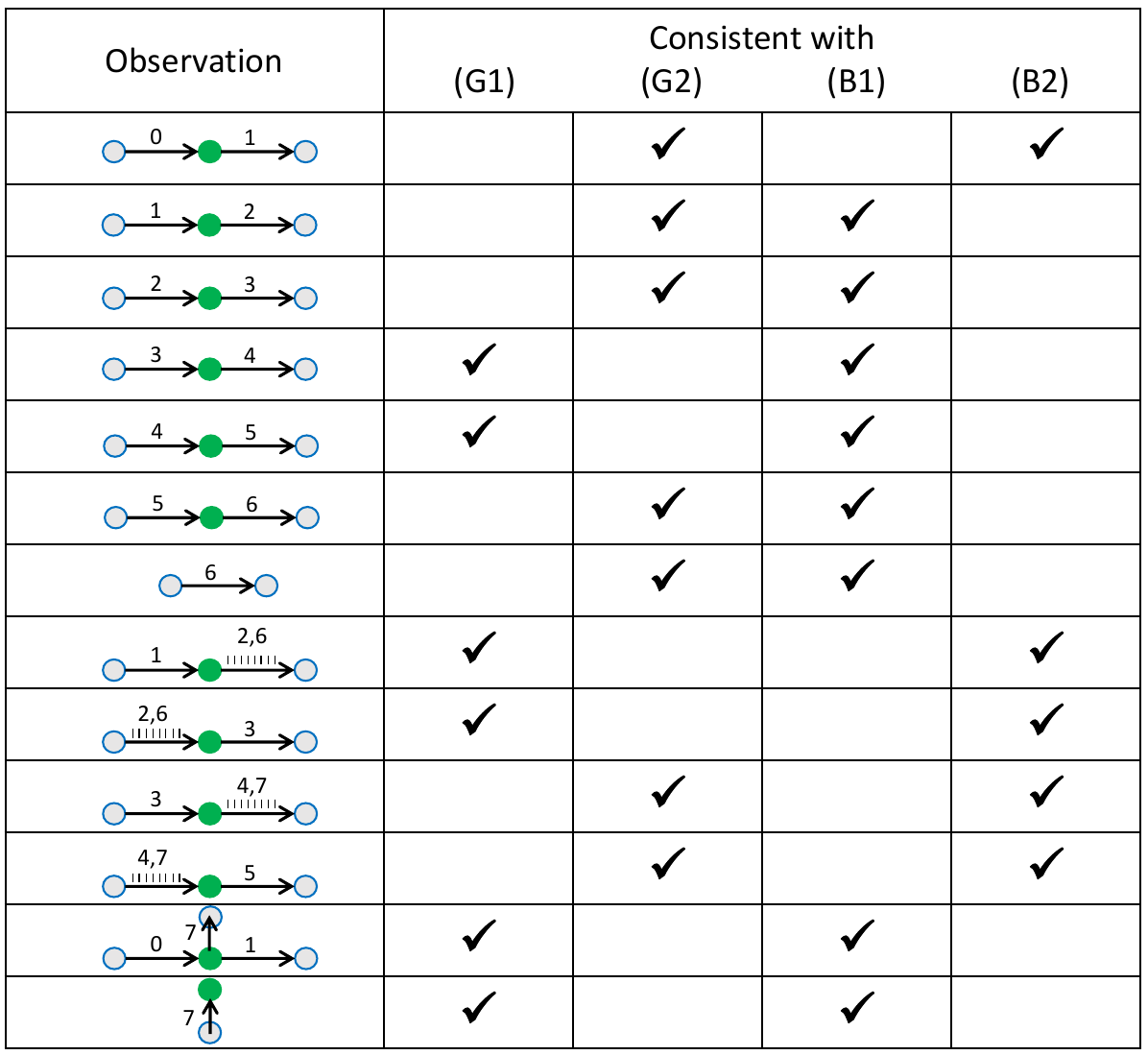}}
\caption{A set of structural observations about an arm that can be revealed in $I_n(v_i)$.}
\label{fig-lb-3}
\end{figure}

%-------------------------------------------------------------------------------------------------------

\section{Reducing degrees in graphs with valid DFS numbering% (Lemma~\ref{lemma:deg-reduction-DFS})
}
\label{section:deg-reduction-DFS}

In this section, we present a useful framework supporting our analysis in \cref{subsec:properties-eps-far}, showing that if a connected undirected graph $G = (V,E)$ of maximum degree $d$ has a valid DFS numbering $\num$, then one can modify $O(k)$ edges of $G$ to obtain a connected graph on vertex set $V$ whose maximum degree is smaller than $d$, where $k$ is the number of vertices with degree $d$. (We observe that it is sufficient to prove the claim for connected graphs only, since if $G$ is not connected, then the same construction can be applied to each connected component individually.)

Let us begin with two straightforward auxiliary claims that follow easily from the fact that the parent relation $p(v)$ defines a DFS tree of a given connected graph equipped with a valid DFS numbering $\num$.

\begin{claim}
\label{claim:DFS-removing-back-edges}
Let $G = (V,E)$ be an arbitrary connected graph equipped with a valid DFS numbering $\num$. Let $\{u,v\}$ be an arbitrary edge in $E$ with $\num(u) < \num(p(v)) < \num(v)$. Then the graph $G' = (V,E \setminus \{u,v\})$ is connected and $\num$ is its valid DFS numbering.
\qed
\end{claim}

\begin{claim}
\label{claim:rewring-the-last-child-DFS}
Let $G = (V,E)$ be an arbitrary connected graph equipped with a valid DFS numbering $\num$. Let $v$ be an arbitrary vertex in $V$ and let $u_1, \dots, u_k$ be all neighbors of $v$ with DFS numbers greater than $\num(v)$, with $\num(u_1) < \dots < \num(u_k)$. Let $w$ be vertex with $\num(w) = \num(u_k)-1$. If $k > 1$ then the graph $G' = (V,E \setminus \{v,u_k\} \cup \{w,u_k\})$ is connected and $\num$ is its valid DFS numbering.

Furthermore, if $k > 1$ then for every neighbor $x$ of $w$ we have $\num(x) < \num(w)$.
\qed
\end{claim}

Using Claims~\ref{claim:DFS-removing-back-edges} and \ref{claim:rewring-the-last-child-DFS}, we can easily prove the following.

\begin{lemma}%[\textbf{One DFS degree reduction}]
\label{lemma:deg-reduction-DFS}
Let $G = (V,E)$ be an arbitrary connected graph of maximum degree $d \ge 3$ equipped with a valid DFS numbering $\num$. Let $V_d$ be the set of vertices of degree $d$ in $G$. Then one can construct a connected graph $G' = (V,E')$ of maximum degree at most $d-1$ for which $\num$ is a valid DFS numbering and such that $|E \triangle E'| \le 3 \cdot |V_d|$.
\end{lemma}

\begin{proof}
We process vertices in $V_d$ one by one, in arbitrary order, and reduce their degrees without increasing the degree of any other vertex to more than $d-1$. For that, iteratively, take a vertex $v \in V_d$.
\begin{itemize}
\item If $v$ has at least two neighbors $u$, $w$ with $\num(u) < \num(w) < \num(v)$ then remove edge $\{u,v\}$ from $G$. By \cref{claim:DFS-removing-back-edges}, the resulting graph is connected with valid DFS numbering $\num$ and has one less vertex of degree $d$.
\item If $v$ has no two neighbors $u$, $w$ with $\num(u) < \num(w) < \num(v)$, then it has $k \ge d-1 \ge 2$ neighbors $u_1, \dots, u_k$ with DFS numbers greater than $\num(v)$, with $\num(u_1) < \dots < \num(u_k)$. Let $w$ be the vertex with $\num(w) = \num(u_k)-1$. Then replace in $G$ edge $\{v,u_k\}$ with a new edge $\{w,u_k\}$. Furthermore, since this increases the degree of vertex $w$, we want to remove one edge incident to $w$ unless $w$ has had initially a unique neighbor. If $w$ had two or more neighbors, then \cref{claim:rewring-the-last-child-DFS} ensures that all of them have DFS numbers smaller than $w$, and hence by \cref{claim:DFS-removing-back-edges}, we can remove one of the incident edges to $w$, maintaining the validity of DFS numbering $\num$ and of the connectivity of the resulting graph.

    By Claims \ref{claim:DFS-removing-back-edges} and \ref{claim:rewring-the-last-child-DFS}, the resulting graph is connected with valid DFS numbering $\num$ and has one less vertex of degree $d$.
\end{itemize}

If we iterative proceed through all vertices from $V_d$, then we obtain a graph $G = (V,E')$ that is connected with valid DFS numbering $\num$ and has no vertex of degree $d$. Since we modify at most three edges per vertex from $V_d$, we conclude that $|E \triangle E'| \le 3 \cdot |V_d|$, as promised.
\end{proof}

%\Artur{We do not use \cref{lemma:deg-reduction-DFS-more-than-1} in the paper and so most likely we will skip it.}
Observe that it is not difficult to generalize \cref{lemma:deg-reduction-DFS} to obtain a bigger reduction of the maximum degree (though this lemma is not needed in our analysis).

\begin{lemma}[\textbf{DFS degree reduction}]
\label{lemma:deg-reduction-DFS-more-than-1}
Let $G = (V,E)$ be an arbitrary connected graph of maximum degree $d$ equipped with a valid DFS numbering $\num$. For any $d'$, let $V_{d'}$ be the set of vertices of degree $d'$ in $G$. Then for any $d^* \ge 3$ we can construct a connected graph $G = (V,E')$ of maximum degree at most $d^*$ for which $\num$ is a valid DFS numbering and such that $|E \triangle E'| \le \sum_{k=d^*+1}^d 3 \cdot (k-d^*) \cdot |V_k|$.
%
%$= \sum_{k=1}^{d-d^*} 3 \cdot k \cdot |V_{d^*+k}| = 3 \cdot |V_{d^*+1}| + 2 \cdot 3 \cdot |V_{d^*+2}| + 3 \cdot 3 \cdot |V_{d^*+3}| + \dots + (d-d^*) \cdot 3 \cdot |V_d|$.
\end{lemma}

\begin{proof}
This follows directly from \cref{lemma:deg-reduction-DFS} by first reducing the maximum degree from $d$ to $d-1$ using at most $3 \cdot |V_d|$ edits of $G$, then reducing the maximum degree from $d-1$ to $d-2$ using at most $3 \cdot |V_d \cup V_{d-1}|$ further edits of $G$, and so on so forth.
\end{proof}

%-------------------------------------------------------------------------------------------------------

} % END OF JUNK

\end{document}